\documentclass{ifacconf}

\usepackage{graphicx}      % include this line if your document contains figures
\usepackage{natbib}        % required for bibliography
\usepackage{graphicx} % Required for inserting images
\usepackage{subcaption}
\usepackage{placeins}
\usepackage{epsfig} % for postscript graphics files
\usepackage{times} % assumes new font selection scheme installed
\usepackage{amsmath} 
\usepackage{amssymb}  
\usepackage{mathtools}
\usepackage{array}
\usepackage{comment}
\usepackage{svg}
\usepackage{empheq}
\usepackage{cases}
\usepackage{enumitem}
\usepackage{mdframed}
\newcommand{\dd}{\mathop{}\!\mathrm{d}}
\newcommand{\s}[1]{\mathrm{#1}}

\makeatletter
\@ifundefined{proof}{%
  \newenvironment{proof}[1][Proof]{%
    \par\noindent\textit{#1.}\ }{%
    \hfill\(\square\)\par}%
}{}
\makeatother
\begin{document}
\begin{frontmatter}

\title{On the Stealth of Unbounded Attacks Under Non-Negative-Kernel Feedback
%The Curse of Integrators
%
%Stealth of Unbounded Attacks in Feedback Systems: The Curse of Integrators and Nonnegative Kernels
%%An Intriguing Tale of Integrals and Integrators %\\ The Curse of the Internal Model Principle%The Curse of the Internal Model Principle: How Unbounded Attacks Manage Stealth
\thanksref{footnoteinfo}} 
%On the Stealth of False Data Injection Attacks against Networked Control Systems containing Chain of Integrators

\thanks[footnoteinfo]{This study is sponsored, in part, by the Swedish Research Council and the Swedish Energy Agency.}

\author[First]{Kamil~Hassan},
\author[First]{Henrik~Sandberg}
%\author[Third]{Third C. Author}

\address[First]{School of Electrical Engineering and Computer Science, KTH Royal Institute of Technology, Sweden (email: kamilha@kth.se, hsan@kth.se)}
%\address[Second]{School of Electrical Engineering and Computer Science, KTH Royal Institute of Technology, Sweden.}
%\address[Third]{Electrical Engineering Department, 
%   Seoul National University, Seoul, Korea, (e-mail: author@snu.ac.kr)}
%%  of interest for the anomaly detector (AD), %ctuation signal from its nominal reference trajectory,  deviation (with respect to its nominal reference trajectory) of the actuation signal
\begin{abstract}                % Abstract of 50--100 words
The stealth of false data injection attacks (FDIAs) against feedback sensors in linear time-varying (LTV) control systems is investigated. In that regard, the following notions of stealth are pursued: For some finite $\epsilon > 0$, i)~an FDIA is deemed $\epsilon$-\emph{stealthy} if the deviation it produces in the signal that is monitored by the anomaly detector %(from its nominal trajectory)
remains $\epsilon$-bounded for all time, %magnitude by which the signal monitored by the anomaly detector to infer attacks deviates from its nominal reference trajectory remains bounded for all time; 
and ii)~the $\epsilon$-stealthy FDIA is further classified as \emph{untraceable} if the bounded deviation dissipates over time (asymptotically). For LTV systems that contain a chain of $q \geq 1$ integrators and feedback controllers with non-negative impulse-response kernels, it is proved that polynomial (in time) FDIA signals %modeled as polynomials (in time) 
of degree $a$ --- growing \emph{unbounded} over time --- will remain 
%unbounded FDIAs that can be modeled as polynomials (in time) with degree $a$ will remain 
i)~$\epsilon \text{-stealthy}$, for some finite $\epsilon > 0$, if $a \leq q$, and ii)~untraceable, if $a < q$. These results are obtained using 
%To obtain these results, 
the theory of linear Volterra integral equations. %is employed. 

%with respect to anomaly detectors that infer attacks by observing the deviation of the actuation signal from its expected nominal trajectory

%non-negative impulse-response kernels, it is proved that if the closed loop contains a chain of $q \geq 1$ integrators, then false data injection attacks (FDIAs) modeled as polynomials (in time) of degree $a$ (that can even grow unbounded over time, for $a \geq 1$) will remain stealthy if $a \leq q$. %These conditions bound the degree $a$ of the FDIA polynomial by a function of the number of integrators in the control loop, 
%thereby characterizing the types of FDIA signals (step, ramp, parabolic, etc.) that can bypass the system’s security to remain stealthy. 
%To obtain these results, the theory of linear Volterra integral equations is employed.
\end{abstract}

\begin{keyword}
Cyber security networked control, Linear parameter-varying systems %, Positive linear systems %, Fault detection and diagnosis %Linear parameter-varying systems, Positive linear systems
\end{keyword}

\end{frontmatter}
%===============================================================================

\section{Introduction} \label{sec: intro}
%Integrators are ubiquitous in control systems. Ranging from their application in controllers (), to their use as an approximate model for robots, .
The growing reliance of modern control systems on IT infrastructure makes them vulnerable to malicious agents and cyber threats, as argued in \cite{hemsley2018history}. To mitigate these threats, it is essential to quantify the underlying vulnerabilities so that appropriate defense strategies can be developed. % can be developed. %This forms the very motivation of the current work, which is devoted to the study of 
In light of this, the current work is devoted to the investigation of the underlying conditions that allow (possibly) unbounded false data injection attacks (FDIAs) on feedback sensors to remain \emph{stealthy}. In particular, we consider feedback systems that contain a chain of integrators in their closed-loop, thereby focusing the scope of the study. 

%the study retains a sense of generality. % we still address a major class of feedback systems in our study. (possibly multiple) integrator(s) (henceforth called a chain of integrators) in the loop, thereby focusing the scope of the study. 

%and the objective is to study the susceptibility of such systems to stealthy FDIAs.

%. To give an example, consider mechanical motion: The mapping of the force applied to a rigid body to its position can be modeled as a double integrator system. This shows that integrators arise naturally in the modeling of many physical processes. Moreover, integral action is also desired in many control applications, to achieve zero steady-state error or for disturbance rejection, for example. Therefore, the study of integrator-endowed systems for their susceptibility to stealthy FDIAs is of practical significance. 
%false data injection attacks (FDIAs). 
%In particular, feedback systems containing controllers with non-negative impulse-response kernels are investigated for their susceptibility to admitting such stealthy FDIAs that are aimed at the sensor channels. %false data injection attacks (FDIAs). %that can even grow unbounded over time. %on the sensor channels. 

%Integrators are ubiquitous in control systems: 
Integrators arise naturally in the modeling of many physical processes.  %From the motion of rigid bodies (where the force applied is mapped to its position through a chain of two single integrators), to the components of an electric circuit (where the current flowing in a capacitive circuit is mapped to the potential difference across the capacitor through a single integrator, for example) are some examples. 
For example, consider the mechanical motion of a rigid body: The mapping of the force applied to its position is modeled as a chain of two integrators (a double integrator system), as studied in \cite{rao2001naive}. %Similarly, the potential difference across an inductor in an electrical circuit is mapped to the current flowing through it through a single integrator. 
Integrators also appear in many control laws, where integral action may be desired to achieve zero steady-state error and disturbance rejection \cite[Chapter~7.10]{franklin2002feedback}. Therefore, the study of such integrator-endowed systems in the context of security has practical significance and implications for real-world systems. With that in mind, the objective of this work is to investigate whether the presence of these integrators in feedback systems makes them susceptible to stealthy FDIAs.
%from a security perspective.
%study of conditions that permit FDIAs to remain stealthy against integrator-endowed systems holds practical significance and real-world implications from a security perspective. %, and relevance from the perspective of the security of real-world systems. 

%considered class of feedback systems hold practical significance, thereby maintaining practical significance and relevance to the security of real-world systems. 

%since they form a major class of feedback systems, studying the susceptibility of integrator-endowed systems to stealthy FDIAs is of practical significance from the point of view of security. %point of view.

The stealth of FDIAs has been investigated in various forms in the current literature. %In each case, its definition seems to depend on the detection scheme employed. 
For example, \cite{pasqualetti2013attack} defines an attack undetected (and, thus, stealthy) if the output measurement under the attack remains indistinguishable from the nominal (no-attack) case. In \cite{teixeira2015secure}, on the other hand, the stealth of an attack is characterized by the dynamic residual that is constructed using both the actuation and the output signals. 
%For example, in \cite{teixeira2015secure}, a residual-based detector is considered, where the residual is modeled as the output of a dynamic system. %with inputs given by the deviation of the actuation signal and the output measurement from their respective nominal trajectories. 
%Correspondingly, an attack is deemed stealthy if the resulting deviation in the actuation signal and the output measurement does not produce a sufficiently large residual to start the alarm. In contrast, \cite{pasqualetti2013attack} define an undetected (and, thus, stealthy) attack as one for which the deviation produced in the output measurement is indistinguishable from the nominal (no-attack) case. 
%These seemingly distinct notions of stealth betray an underlying principle: The detection (and, thus, stealth) of an FDIA is inferred by observing a deviation signal, where an observed trajectory of the control system is compared with its reference (nominal) trajectory.
While seemingly distinct, in each of these cases, the stealth of the FDIA is based on a measure of the deviation it produces in a signal of interest to the anomaly detector (AD) employed.
%a signal of interest for the respective detection scheme applied. 
In view of this, in this work, we define the stealth of an FDIA in terms of the maximum (finite) amplitude, $\epsilon > 0$, by which the actuation signal deviates (from its nominal trajectory) in response to this FDIA. This characteristic is then termed the $\epsilon \text{-stealth}$ of the FDIA. In addition, if this deviation vanishes over time (asymptotically), we deem the FDIA \emph{untraceably stealthy (u.s.)}. %, where, in addition to being $\epsilon \text{-bounded}$, the deviation generated by the FDIA in the actuation signal vanishes asymptotically. %, we consider the FDIA to have attained \emph{untraceable stealth}. %Furthermore, for attack detection, the reference (nominal) trajectories need not correspond to the (constant) equilibrium of the system. To give an example, consider a satellite in its orbit, following a circular trajectory. The modeling of its local behavior around the nominal circular trajectory then yields a linear time-varying description of its dynamics. 
%In consequence, to study the stealth of an FDIA, the resulting deviation of the actuation signal from its nominal reference trajectory needs to be modeled. 
%In consequence, to determine the stealth of an FDIA, the evolution of the deviation of the actuation signal from its nominal reference trajectory needs to be examined. 
%The reference nominal trajectory of the actuation signal is known 
Under this framework, we study the stealth of FDIAs against integrator-endowed feedback systems by examining the deviation of the actuation signal with respect to its nominal reference trajectory.

%With the objective of developing a mathematical framework for analyzing stealth, we model the deviation of the actuation signal (from its nominal course) as the output of the linearized controller. 
To account for non-constant nominal trajectories in our study, %framework to analyze stealth, 
we adopt linear time-varying (LTV) representations for both the plant and the controller around these trajectories. This permits us to extend our results to study the stealth of attacks against satellites, for example, where the nominal reference trajectory is represented by their elliptical orbit (and not by a constant equilibrium of the system). Then, corresponding to the setup described above, we model the deviation of the actuation signal with respect to its nominal trajectory as a linear Volterra integral equation (LVIE), thereby transforming the stealth analysis into an investigation of the conditions that render the LVIE stable. 
In doing so, we make the following contributions: % in this study:
\begin{enumerate}
\item \textbf{(characterization of stealthy FDIAs):} For LTV systems that contain a chain of integrators and feedback controllers with \textbf{\emph{non-negative}} impulse-response kernels, we identify (by their polynomial degree) the types of unbounded FDIA signals that will remain stealthy to an AD comparing the actuation signal with its nominal trajectory.
%We quantify the \textcolor{red}{susceptibility} of feedback systems with a chain of integrators in their closed-loop to \textcolor{red}{polynomial-type FDIA signals in terms of the maximum degree of such FDIAs that can remain stealthy with respect to the actuation signal.} 
%We quantify the risk to security integrators in the closed-loop poses for feedback systems by relating the number of integrators in the loop and the degree of a polynomial FDIA that can manage stealth against it. 
%show that the presence of integrators in feedback systems can make them susceptible to stealthy FDIAs. % risk to security for feedback systems.
\item \textbf{($\epsilon \text{-stealth}$):} For non-negative-kernel feedback systems with $q \geq 1$ integrators in their closed-loop, it is proved that polynomial (in time) FDIA signals of degree $a$ will remain $\epsilon \text{-stealthy}$, for some finite $\epsilon > 0$, if $a \leq q$.
\item \textbf{(untraceable stealth):} For these non-negative-kernel feedback systems, we further prove that if $a < q$, then the polynomial FDIA signals will attain untraceable stealth.
\item \textbf{(unproven extensions):} Through numerical simulations, we show that the previously stated theoretical results pertaining to the stealth of the polynomial-type FDIAs may also extend to the case of \textbf{\emph{non-positive}}-kernel feedback systems with $q \geq 1$ integrators in the closed-loop. % to both the actuation and the output measurement signals. 
%\item 
%To the best of our knowledge, the security implications of having integrators in the 
\end{enumerate}

The organization of the remaining article is as follows: In Section~\ref{sec: prelim LVIE}, a brief background is provided on linear Volterra integral equations. In Section~\ref{sec: sys des}, first, the setup is described, and then the problem is formulated. In Section~\ref{sec: main results}, the main results related to the stealth of FDIAs are presented, which is then followed by the numerical simulations in Section~\ref{sec: num ex}. Finally, concluding remarks are included in Section~\ref{sec: concl}. 

%The remainder of the article is organized as follows: In Section~\ref{sec: prelim LVIE}, a brief background on linear Volterra integral equations is provided. %, where some classical results pertaining to their stability are presented.
%Subsequently, a description of the considered setup is expatiated in Section~\ref{sec: sys des}, which is followed by the formulation of the problem statement. The results pertaining to the stealth of FDIAs are then presented in Section~\ref{sec: main results}. In Section~\ref{sec: num ex} theoretical results are substantiated via numerical examples, and, finally, concluding remarks are given in Section~\ref{sec: concl}.

\emph{Notations:} The set of real numbers is denoted by $\mathbb{R}$, and the set of integers by $\mathbb{Z}$. Furthermore, for $r \in \mathbb{R}$ and $s \in \mathbb{Z}$, $\mathbb{R}_{\geq r} \coloneqq \{x \in \mathbb{R} : x \geq r\}$ and $\mathbb{Z}_{\geq s}\coloneqq \{x \in \mathbb{Z}: x \geq s\}$. Finally, for some interval~$\mathcal{I} \subseteq \mathbb{R}_{\geq 0}$ and set $\mathcal{J} \subseteq \mathbb{R}^{m}$, $\mathsf{C}(\mathcal{I}, \mathcal{J})$ denotes the set of continuous functions mapping $\mathcal{I}$ to $\mathcal{J}$.

\section{Preliminaries on Linear Volterra Integral Equations} \label{sec: prelim LVIE}
The following (scalar) integral equation is characterized in the literature as %Consider the following scalar 
a non-homogeneous linear Volterra integral equation (LVIE) of the second kind %type~2
\begin{equation}    \label{eq: generic LVIE}
    x(t) = \int_0^t \, G(t,\tau)\, x(\tau) \, \dd \tau + \phi(t), \quad t\geq 0,
\end{equation}
where $x(t)\in \mathbb{R}$, the input $\phi(t) \in \mathbb{R}$, and the causal kernel $G:\mathbb{R}_{\geq 0}\times \mathbb{R}_{\geq 0} \to \mathbb{R}$. % such that $G(t,\tau) = 0$ for~$\tau> t$. 
See, for example, \cite{gripenberg1990volterra} and \cite{burton2005volterra} for some references on LVIEs. %For \eqref{eq: generic LVIE}, w
We assume that for any $T > 0$, $G$ in \eqref{eq: generic LVIE} is continuous on the domain~$\{(t,\tau) \in \mathbb{R}_{\geq 0} \times \mathbb{R}_{\geq 0} : 0 \leq t \leq T,0 \leq \tau \leq t \}$. This ensures that, corresponding to any continuous input~$\phi: \mathbb{R}_{\geq 0} \to \mathbb{R}$, there exists a unique solution~$x: \mathbb{R}_{\geq 0} \to \mathbb{R}$ to \eqref{eq: generic LVIE} \cite[Theorem~1.2.3]{brunner2017volterra}. Next, we define the different notions of stability as they pertain to the zero solution of \eqref{eq: generic LVIE} (that is, $x(t) \equiv 0$ corresponding to $\phi(t) \equiv 0$).
\begin{defn} \label{def: Stability ZB}
    The zero solution of \eqref{eq: generic LVIE}  is \emph{\textbf{stable}} if for all~$\epsilon > 0$, there exists~$\delta(\epsilon) >0$ such that~$$\sup_{t\geq 0} |\phi(t)|< \delta(\epsilon) \implies \sup_{t \geq 0}|x(t)| < \epsilon.$$
\end{defn}

\begin{defn} \label{def: AS ZB}
       The zero solution of \eqref{eq: generic LVIE} is \emph{\textbf{asymptotically stable}} if it is stable and there exists~$\delta_1 > 0$ such that 
    $$\sup_{t \geq 0} |\phi(t)|  < \delta_1 \text{ and } \lim_{t \to \infty}\phi(t)=0  \implies \lim_{t \to \infty}x(t) = 0.$$
\end{defn}

Now we present some results from the existing literature to establish the sufficient and necessary conditions needed for \eqref{eq: generic LVIE} to be stable as per the notions described in Definitions~\ref{def: Stability ZB} and \ref{def: AS ZB}. %Note that the following results assume~$G(t,s) \geq 0$, for all $0 \leq s \leq t$.

\begin{lem}[\cite{tsalyuk1979volterra},~Section~1.3.1] \label{lem: iff stable ZB}
Let $G(t,\tau) \geq 0$, for all $0 \leq \tau \leq t < \infty$. The zero solution of \eqref{eq: generic LVIE} is stable per Definition~\ref{def: Stability ZB} if, and only if,
    \begin{align} \label{eq: iff bounded kernel cond}
        \sup_{t \geq 0} \int _0^t G(t,\tau)\, \dd \tau < \infty, 
    \end{align}
    and, for some~$v \in \mathbb{Z}_{\geq 1}$, %the spectrum of the matrix
    \begin{equation} \label{eq: iff iterate cond}
    \begin{split}        
        A_v  %\lim_{T\to \infty} \, \limsup_{t \to \infty} \, \int_{T}^t \, G_v(t,\tau) \, \dd \tau, \\ 
        = \lim_{T \to \infty} \sup_{t \geq T} \int _{T}^{t} \, G_v(t,\tau)\, \dd \tau
    \end{split}
    \end{equation}
    lies within the unit circle, where~$G_v(t,\tau)$ is the~$v^{\text{th}}$ iterate of the kernel~$G(t,\tau)$ such that $G_1(t,\tau) \equiv G(t,\tau)$ and
    \begin{equation}
        G_{m}(t,\tau) \coloneqq \int_\tau^t G(t,\sigma) G_{m-1}(\sigma , \tau) d\sigma, \quad m \in \mathbb{Z}_{\geq 2}. 
    \end{equation}
    %\footnote{Admittedly, we cannot parse \eqref{eq: iff iterate cond} at the moment. However, for our purposes, it is not required to focus on this condition and it has only been included for completion. Nonetheless, we shall revisit it at a later stage.}.
\end{lem}
\begin{lem}[\cite{tsalyuk1979volterra},~Section~1.3.3] \label{lem: iff AS ZB}
    Let $G(t,\tau) \geq 0$, for all $0 \leq \tau \leq t< \infty$. The zero solution of \eqref{eq: generic LVIE} is asymptotically stable per Definition~\ref{def: AS ZB} if, and only if, it is stable per Definition~\ref{def: Stability ZB}, and, for any~$T > 0$,
        \begin{equation} \label{eq: iff AS ZB}
            \lim_{t \to \infty} \, \int_{0}^T\, G(t,\tau) \, \dd \tau = 0. 
        \end{equation}
%    \end{enumerate}   
\end{lem}
\begin{rem}    
Lemmas~\ref{lem: iff stable ZB} and~\ref{lem: iff AS ZB} apply to LVIEs of the form \eqref{eq: generic LVIE} if they have a non-negative kernel $G(t,\tau) \geq 0$ for all $0 \leq \tau < t < \infty$. %Systems that satisfy this property are henceforth referred to as positive systems in this study. (except for the necessity component)
Admittedly, this restricts the applicability of these results (and our ensuing analysis that uses them). We can, however, extend these results to the instantiations of \eqref{eq: generic LVIE} with non-positive kernels: As pointed out at the beginning of \cite[Section~1.3]{tsalyuk1979volterra}, %if \eqref{eq: generic LVIE} with kernel $\bar{G}(t,s) \equiv |G(t,\tau)|$ is stable per Definition~\ref{def: Stability ZB} (asymptotically stable per Definition~\ref{def: AS ZB}) then so is \eqref{eq: generic LVIE} with the kernel $G(t,\tau)$. That is, 
if the conditions stated in Lemmas~\ref{lem: iff stable ZB} and~\ref{lem: iff AS ZB} are satisfied for the kernel $|G(t,\tau)| \in \mathbb{R}_{\geq 0}$ instead, then it is sufficient to imply that the zero solution of \eqref{eq: generic LVIE} (with the possibly non-positive kernel $G(t,\tau) \in \mathbb{R}$) is stable per Definitions~\ref{def: Stability ZB} and~\ref{def: AS ZB}, respectively (see Example~2 in Section~\ref{sec: num ex}).  %However, for LVIEs with non-negative kernels, the conditions in Lemmas~\ref{lem: iff stable ZB} and~\ref{lem: iff AS ZB} are both sufficient and necessary
%for their stability per Definitions~\ref{def: Stability ZB} and~\ref{def: AS ZB}, respectively. %, and thus, synonymous with their stability.
\end{rem}
%unavoidable with regard to their stability.   

\begin{comment}    
Finally, we state the following Lemma, which we will refer to in the next section. 
\begin{lem}[\cite{Tsa68},~Lemma~5] \label{lem: pre iff AS}
    Suppose~$G(t,s) > 0$, and let \eqref{eq: iff bounded kernel cond} be true. Then, for any continuous function $ \lim_{t \to \infty} x(t) = 0$, the limit 
    \begin{equation}
        \lim_{t \to \infty} \, \int_0^t \, G(t,s) \, x(s)\, ds = 0,
    \end{equation}
    holds if, and only if, for any~$T>0$
    \begin{equation}
        \lim_{t\to \infty}\, \int_0^T\, G(t,s) \, ds = 0. 
    \end{equation}
\end{lem}
\end{comment}

%In the subsequent subsection, we shall employ the concepts and results discussed here to address Problem~\ref{prob: prob}.

\section{System Description and Problem Formulation} \label{sec: sys des}
In this section, we endeavor to provide context for the current study. To that end, first, the setup is described, and then the description is employed towards formulating the problem. 
%we describe the setup considered and then, subsequently, employ the description towards formulating the problem addressed.  

\begin{figure}[t!]
    \centering
    \includegraphics[scale = 0.13]{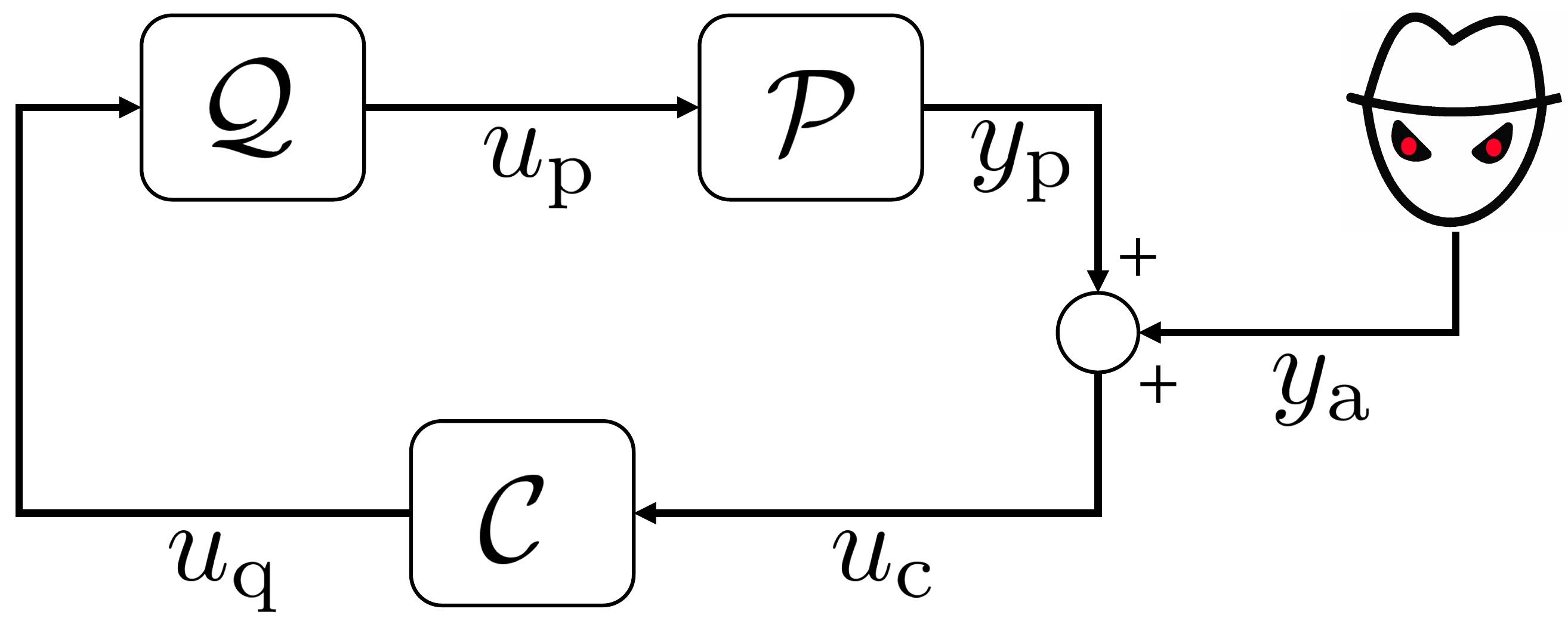} %[scale = 0.5]
    \caption{Control system under false data injection attack~$y_{\s{a}}$.  
    \label{fig: system}}
\end{figure}

\subsection{Setup} \label{sec: setup}
Consider the single-input--single-output (SISO) system depicted in Fig.~\ref{fig: system}.  Except for the plant output~$y_{\mathrm{p}}$ (and the exogenous input~$y_{\s{a}}$), every signal in the depicted control system is labeled such that it is indicative of which subsystem it forms an input to. For example,~$u_{\mathrm{c}}$ represents the input to the controller~$\mathcal{C}$. Similarly, the output of the controller is labeled $u_{\mathrm{q}}$ because it is the input to the integrator system~$\mathcal{Q}$, etc. Keeping this in mind, the constituent subsystems of the considered control system in Fig.~\ref{fig: system} are described as follows:
%the composition of the control system is detailed as follows:
\begin{itemize}
    \item The linearized dynamics of the plant~$\mathcal{P}$ and the controller~$\mathcal{C}$ (around some nominal trajectories) are represented by the following linear time-varying (LTV) formulation
    \begin{equation} \label{eq: LTV dyn}
        \begin{split}
            \dot{x}_{\s{i}}(t) &= A_{\s{i}}(t) x_{\s{i}}(t) + B_{\s{i}}(t)u_{\s{i}}(t), \quad x_{\s{i}}(0) = \mathbf{0},  \\
            z_{\s{i}}(t) &= C_{\s{i}}(t) x_{\s{i}}(t),
        \end{split}
    \end{equation}
    where $x_{\mathrm{i}} \in \mathbb{R}^{n_{\mathrm{i}}}$, $A_{\s{i}} \in \mathsf{C}(\mathbb{R}_{\geq 0}, \mathbb{R}^{n_{\s{i}} \times n_{\s{i}}})$, $B_{\s{i}} \in \mathsf{C} (\mathbb{R}_{\geq 0}, \mathbb{R}^{n_{\s{i}}})$, $C_{\s{i}} \in \mathsf{C} (\mathbb{R}_{\geq 0}, \mathbb{R}^{1 \times n_{\s{i}}})$, $u_{\mathrm{i}} \in \mathbb{R}$, and the indicator subscript $\s{i} \in \{\s{p} , \s{c}\}$
    such that $\s{i} = \s{p}$ implies that the variable corresponds to the plant and $\s{i} = \s{c}$ means that it corresponds to the controller. Furthermore, for $\s{i} \in \{\s{p}, \s{c}\}$, the placeholder variable $z_{\mathrm{i}} \in \mathbb{R}$ denotes the output of \eqref{eq: LTV dyn} such that $z_{\s{p}} \equiv y_{\s{p}}$ (corresponding to the plant) and $ z_{\s{c}} \equiv u_{\s{q}}$ (corresponding to the controller). 

    %Now, we employ the standard linear systems theory (see [], for example).
    Moving forward, to model the input--output mappings of these subsystems, we define the impulse-response kernel $g_{\s{i}}:\mathbb{R}_{\geq 0} \to \mathbb{R}_{\geq 0} \to \mathbb{R}$ associated with \eqref{eq: LTV dyn} as follows
\begin{equation} \label{eq: g intro}
    g_{\s{i}}(t,\tau) \coloneqq C_{\s{i}}(t)\,\Phi_{\s{i}}(t,\tau) \, B_{\s{i}}(\tau), %\quad 0 \leq \tau \leq t < \infty,
\end{equation}
where the state-transition matrix $\Phi_{\s{i}}:\mathbb{R}_{\geq 0} \times \mathbb{R}_{\geq 0} \to \mathbb{R}^{n_{\s{i}} \times n_{\s{i}}}$ is given by the unique solution of 
\begin{equation} \label{eq: Phi}
    \frac{\partial}{\partial t} \Phi_{\s{i}}(t,\tau) = A_{\s{i}}(t) \, \Phi_{\s{i}}(t,\tau), \quad \Phi_{\s{i}}(\tau,\tau) = I, %\quad t\geq 0,  
\end{equation}
%and $A_{\s{i}}$, $B_{\s{i}}$ and $C_{\s{i}}$ are as in \eqref{eq: LTV dyn}, 
for $\s{i} \in \{\s{p}, \s{c}\}$. Then, the input--output model for \eqref{eq: LTV dyn} is formulated as %is allows us to represent the output of the plant using the following convolution-like integral:
\begin{align} %\label{eq: var par y}
     z_{\s{i}}(t) &= \int_0^t g_{\s{i}}(t,\tau) \,  u_{\s{i}}(\tau)\, \dd \tau, \quad t \geq 0, \label{eq: plant g_p}
\end{align}
where $\s{i}\in \{\s{p}, \s{c}\}$, $z_{\s{p}} \equiv y_{\s{p}}$, and $z_{c} \equiv u_{q}$, as described in \eqref{eq: LTV dyn}. Note that the above modeling is standard in linear systems theory. For further details in this regard, see, for example, \cite{rugh1996linear}.
%A similar representation of the controller is also obtained:
%\begin{align} 
%     u_{\s{q}}(t) &= \int_0 ^t g_{\s{c}}(t,\tau) \, u_{c}(\tau) \, \dd \tau, \quad t \geq 0. \label{eq: contr g_c}
%\end{align} 

    \item The subsystem~$\mathcal{Q}$ represents a chain of $q \in \mathbb{Z}_{\geq 1}$ single-integrators with dynamics modeled as
    \begin{equation} \label{eq: SISO integrator}
        \begin{split}
            \dot{x}_{\mathrm{q}} & = \underbrace{\begin{bmatrix}
                0 & 1 & 0 & \cdots & 0 \\
         0 & 0 & 1 & \cdots & 0 \\
         \vdots & \vdots & \ddots& \ddots & \vdots\\
         0 & 0 & 0 &\cdots &1\\
         0 & 0 & 0 & \cdots & 0
            \end{bmatrix}}_{\eqqcolon {A}_{\s{q}}} x_{\mathrm{q}} + \underbrace{\begin{bmatrix}
                0 \\
                0 \\
                \vdots \\
                0\\
                1
            \end{bmatrix}}_{\eqqcolon B_{\s{q}} } u_{\mathrm{q}}, \\
                u_{\mathrm{p}} &= \, \underbrace{\begin{bmatrix}
                1 & 0 & 0& \cdots & 0
            \end{bmatrix}}_{\eqqcolon C_{\s{q}}} x_{\mathrm{\s{q}}},
        \end{split}
    \end{equation}
    where $x_{\s{q}} \in \mathbb{R}^q$, $A_\s{q} \in \mathbb{R}^{q\times q}$, $B_{q}\in \mathbb{R}^q$ and $C_{\s{q}}\in \mathbb{R}^{1 \times q}$. Since these dynamics are linear time-invariant (LTI), the state-transition matrix corresponding to \eqref{eq: SISO integrator} takes the form $\Phi_{\s{q}}(t) \coloneqq e^{A_{\s{q}}t}$, where $A_{\s{q}}$ is as given in \eqref{eq: SISO integrator}. Consequently, the input--output model of $\mathcal{Q}$ is given %can be formulated via 
    by the following convolution integral equation: %-based input--output model ensues
\begin{equation} \label{eq: up gq}
        u_{\s{p}}(t) = \int_0^t g_{\s{q}}(t-\tau) u_{q}(\tau)\, \dd \tau, \quad t \geq 0,
\end{equation}    
where\begin{align} \label{eq: gq}
    g_{\s{q}}(t) &\coloneqq C_{\s{q}} \Phi_{\s{q}}(t) B_{\s{q}} = \frac{t^{q-1}}{(q-1)!}, 
\end{align}
$C_{\s{q}}$ and $B_{\s{q}}$ as in \eqref{eq: SISO integrator}, and $q \in \mathbb{Z}_{\geq 1}$. 

\end{itemize}

The LTV representations of~$\mathcal{P}$ and $\mathcal{C}$ described in \eqref{eq: LTV dyn} correspond to some non-linear dynamics that are linearized around the nominal trajectories followed by the system. Therefore, under normal conditions ($y_{\s{a}} \equiv 0$), the closed-loop system is expected to follow these nominal trajectories, ensuring the deviations $u_{\s{q}}, \, u_{\s{c}}, \, u_{\s{p}}$, and $y_{\mathrm{p}} \equiv 0$ (neglecting the initial convergence period). In the next subsection, we account for the input~$y_{\s{a}}$ in our model, and formulate the problem studied. %addressed in this study. 

\subsection{Problem Statement} \label{sec: LTV prob formulation}
In this section, we explicate the role of the exogenous input~$y_{\s{a}}$ --- henceforth termed the false data injection attack (FDIA) --- as it pertains to the problem considered in this study. %In particular, owing to the presence of the integrator system $\mathcal{Q}$ in our setup, the FDIA-signal~$y_{\s{a}}$ is modeled 
The FDIA is modeled as the output of a chain of $a \in \mathbb{Z}_{\geq 0}$ integrators such that %follows: 
%the output of a chain of $a \in \mathbb{Z}_{\geq 0}$ integrators with a constant input $\tilde{h}(t) \equiv h \in \mathbb{R}$, $t \geq 0$. In accordance with the modeling of $\mathcal{Q}$, the input--output representation of the FDIA is formulated as
%with a constant input $\tilde{h}(t) \equiv h \in \mathbb{R}$, $t \geq 0$. In accordance with the modeling of $\mathcal{Q}$ in \eqref{eq: up gq}, the input--output model for the FDIA is then formulated as
\begin{align} \label{eq: FDIA LTV} 
    y_{\s{a}}(t) &\coloneqq \int_0^t \underbrace{\frac{(t-\tau)^{a-1}}{(a-1)!}}_{\eqqcolon g_{\s{a}}(t-\tau)}\,\tilde{h}(\tau) \, \dd \tau = \frac{h}{a!} t^a, \quad t \geq 0,
\end{align}
where $\tilde{h}(t) \equiv h \in \mathbb{R}$ denotes the constant input to the chain. In effect, the FDIA~$y_{\s{a}}$ represents a polynomial (monomial) of degree~$a \in \mathbb{Z}_{\geq 0}$, which grows unbounded over time (for $a \geq 1$) with a rate proportional to $h$. Correspondingly, in response to this FDIA, we posit that the system can be led away from its nominal course and rendered unable to meet its given objective. 
%This signal, when added to the sensor channel as shown in Fig.~\ref{fig: system}, corrupts the system measurements received by the controller, which can then render the system unable to meet its given objective. 
%thereby (possibly) steering the system away from its nominal course and rendering it unable to meet its given objective. 
%, and, consequently, the risk to the security it poses for the LTV system is quantified in terms of its stealth. However, before delving into that, we make a note that 
%However, as shown in Fig.~\ref{fig: system}, an FDIA~$y_{\s{a}} \in \mathbb{R}$ on the sensor channel can corrupt the system measurement received by the controller as feedback, thereby (possibly) rendering the system unable to meet its given objective (follow the nominal trajectories). 

To mitigate the effects of such FDIAs and ensure the sanctity of the feedback loop, anomaly detectors (ADs) are installed to detect deviations of the signals in the control loop from their nominal trajectories. %In that regard, of particular interest are the inputs and outputs of the controller: $u_{\s{c}}$ and $u_{\s{q}}$, respectively, due to their ease of access. %of the system from nominal trajectories. 
To counteract such security measures, the deviations produced in these signals by the FDIA must be small enough to bypass these detectors.
%must elicit minimal effect on these signals so that it can remain stealthy and bypass these detectors. %, thereby essentially achieving stealth with respect to them. 
This evasion is crucial from the perspective of the attacker, for the measure of a successful attack is not the harm it can cause but rather the harm it can cause while remaining undetected\footnote{Once the attack is detected, in most cases, the system can simply be shut down if worst comes to worst, thus limiting the impact of the attack.}. % , thereby potentially being able to mitigate the attack encountered.}. 
From a security perspective, therefore, it is imperative to understand the underlying vulnerabilities of the system that would enable such FDIAs to remain stealthy. % for a proper analysis of the risks involved.

%From the perspective of defense, it is imperative to be aware of the vulnerabilities of the system, and the conditions that may permit such stealthy FDIAs. %cannot penetrate through the security system security. 
%that such FDIAs are not allowed stealth. 
%improving the security of the system should then be prioritized to detect such adversarial injections.  

In view of the above discussion, %this study is devoted to the investigation of the conditions under which FDIAs of the form \eqref{eq: FDIA LTV} can attain stealth against the considered setup.  
%the vulnerabilities of the considered system as they pertain to permitting stealth for the FDIA in \eqref{eq: FDIA LTV}. 
%To that end, 
we note that the \emph{stealth} of an FDIA is characterized by its ability to remain undetected in its effect on the signal(s) of interest for the AD. In the context of this study, owing to its ease of access and availability for monitoring, $u_{\s{q}}$ (the output of the controller) is deemed such a signal of interest. %That is, by observing $u_{\s{q}}$, an inference regarding the departure of system operation from nominal conditions is made by observing. 
That is, we posit that by observing $u_{\s{q}}$, an inference could be made regarding the departure of system operation from nominal conditions. 
%Then stealth is defined as follows:
Correspondingly, the following notion of stealth is presented: % for the FDIA in \eqref{eq: FDIA LTV}: 
%In particular, we are interested in quantifying the  of the considered system  allow for such stealthy FDIAs. However, to commence, first, we need to formalize the notion of stealth for FDIAs as follows:
%To formulate the problem more clearly, we formalize the notion of stealth as follows:
%We formalize the notion in the following manner:% and we formalize the notion in the following way% define it as follows:
\begin{defn} \label{def: stealth}
For the setup considered in Fig.~\ref{fig: system}, the FDIA~$y_{\s{a}} \in \mathbb{R}$ is said to be: %described in \eqref{eq: FDIA LTV} is deemed:
\begin{enumerate}
    \item \textbf{\emph{$\epsilon \text{-stealthy}$}} with respect to $u_\s{\s{q}}$ if $$\sup_{t \geq 0}|u_{\s{q}}(t)| \leq \epsilon,$$ for some~$\epsilon > 0$. %, $\s{i} \in \{\s{q}, \s{c}\}$. 
    \item \textbf{\emph{untraceably stealthy (u.s.)}} with respect to $u_{\s{q}}$ if it is $\epsilon \text{-stealthy}$ with respect to $u_{\s{q}}$ \emph{and}  
    $$\lim_{t \to \infty} |u_{\s{q}}(t)| = 0.$$
%    for $\s{i} \in \{\s{q} , \s{c}\}$.
\end{enumerate}
\end{defn}%vulnerability 
%In the present study, we investigate the system considered in Fig.~\ref{fig: system}  for its vulnerability in admitting such stealthy FDIAs. %The setup is atypical because it contains a chain of integrators %(represented by the $\mathcal{Q}$-integrator system in Fig.~\ref{fig: system}) 
%in its closed-loop.
%In particular,  
%In particular:
%\vspace{0.25 cm}
%\begin{mdframed}
%We aim to characterize the FDIAs that could manage stealth with respect to the residual signals of the considered control system per Definition~\ref{def: stealth}. %, and establish a link (if it exists) between the number of integrators in the closed-loop system and the system that generates such stealthy attacks. %that manages 
%\end{mdframed}
%In the sequel, we shall show that the considered setup permits such a characterization (under certain conditions) owing to the inclusion of a chain of integrators in its \textcolor{red}{stable} closed-loop. %This is investigated next for the LTI case. 

%Based on this definition, we shall focus on the modeling of~$u_{\s{q}}$ in the next section, before formulating the problem statement. 
This leads to %us to formulate 
the following problem statement:
\vspace{0.2cm}
\begin{mdframed}
\begin{prob} \label{prob: prob}    
For a given~$a \in \mathbb{Z}_{\geq 0}$, does there exist $h \neq 0$ such that the corresponding FDIA $y_{\s{a}}$ in \eqref{eq: FDIA LTV} remains stealthy with respect to $u_\s{q}$ per Definition~\ref{def: stealth}?
%Under what conditions (if any) can the FDIA $y_{\s{a}}$ modeled as in \eqref{eq: FDIA LTV} remain stealthy with respect to $u_\s{q}$ (the output of the controller), per Definition~\ref{def: stealth}?  
\end{prob}
\end{mdframed}
\vspace{0.2cm}
%Given the LVIE-based representation of $u_{\s{q}}$ in \eqref{eq: uq volt}, under what conditions does the FDIA modeled in \eqref{eq: FDIA LTV} attain stealth with respect to $u_{\s{q}}$, per Definition~\ref{def: stealth}?
In other words, we shall investigate the conditions (if there are any) that allow the FDIA modeled in \eqref{eq: FDIA LTV} --- growing unbounded over time for $a \geq 1$ --- to manage stealth against the considered setup.  %per Definition~\ref{def: stealth}. 
In the next section, we employ the theory of LVIEs from Section~\ref{sec: prelim LVIE} to address Problem~\ref{prob: prob}.

%show that there indeed exists a condition, linking the number of integrators~$q$ in the closed-loop and the number of integrators~$a$ used in the construction of \eqref{eq: FDIA LTV}, that may\footnote{Under certain assumptions (see Theorems~\ref{thm: bounded attack} and~\ref{thm: u.s.}).} permit such unbounded FDIAs to remain stealthy. %(under certain assumptions). 
%that effect, and it is shown
%\begin{rem}  \label{rem: intuition} %, with the %potential of growing unbounded over time (when the degree of the FDIA polynomial $a \geq 1$),  
%The considered model of the FDIA in \eqref{eq: %FDIA LTV} %(polynomial with degree~$a\in %\mathbb{Z}_{\geq 0}$) %, which can grow unbounded %over time if $a \geq 1$) 
%is chosen %of interest in relation to %Problem~\ref{prob: prob} 
%because of the following intuition: If the LTV feedback controller can stabilize its output despite having a series of $q\in \mathbb{Z}_{\geq 1}$ integrators in the closed-loop, then, by linearity, it should also bound the contribution of the FDIA in its output, which, in itself, is modeled as the output of a series of $a$ integrators, if $a \leq q$. 
%as it is generated from a series of $a \in \mathbb{Z}_{\geq 0}$ integrators, given $a \leq q$. 
%In the next section, we employ the theory of LVIEs presented in Section~\ref{sec: prelim LVIE} to derive the conditions (if there are any) under which the intuition can be actualized. 
%, and under what conditions, the intuition is actualized, and if yes, then under what conditions? %investigate whether that is, in fact, the case. 
%\end{rem}

\section{Stealth of False Data Injection Attacks Against Linear Time-Varying Systems} \label{sec: main results}
In this section, we address Problem~\ref{prob: prob}. To that end, first, we represent the signal of interest $u_{\s{q}}$ as an LVIE of the form \eqref{eq: generic LVIE} and then, leveraging this mathematical formulation, we develop the theory of stealth for the FDIA \eqref{eq: FDIA LTV} in relation with Problem~\ref{prob: prob}.

%we provide a brief background on the stability of LVIEs, where we report some results from the classical literature to establish a foundation. Then, using these tools, we develop the theory of stealth as it applies to Problem~\ref{prob: prob}.

\subsection{LVIE-based Representation of $u_{\s{q}}$} 
In this subsection, we model $u_{\s{q}}$ as an LVIE. %we seek a modeling framework that quantifies the impact of the FDIA~$y_{\s{a}}$ on $u_{\s{q}}$. 
However, before proceeding, we state the following result. %integrability condition on the impulse-response kernels of the plant and the controller. %(\emph{Integrability}) 
\begin{lem}\label{lem: g cont} 
%Given the LTV systems described in \eqref{eq: LTV dyn}, and the corresponding kernels~$g_{\s{p}}$ and $g_{\s{c}}$ as given in \eqref{eq: plant g_p}, for $ \s{i} = \s{p}$ and $\s{i} = \s{c}$, respectively. 
For any~$T > 0$, the kernels $g_{\s{p}}(t,\tau)$ and $g_{\s{c}}(t,\tau)$, defined as in \eqref{eq: g intro} %(in association with \eqref{eq: LTV dyn}), 
(for $ \s{i} = \s{p}$ and $\s{i} = \s{c}$, respectively), are continuous over the domain~$\mathcal{D}_T \coloneqq \{(t,\tau) \in \mathbb{R}_{\geq 0} \times \mathbb{R}_{\geq 0}: 0 \leq t \leq T, \, 0 \leq \tau \leq t\}$. 
%    \item $g_\s{c}(t,t) \neq 0$, for all~$t$.
\end{lem}
\begin{proof}
%The proof is sketched as follows: 
%The matrix-valued mappings~$A_{\s{i}}$, $B_{\s{i}}$, and~$C_{\s{i}}$ considered in \eqref{eq: LTV dyn} are continuous over the compact interval $[0,T]$, for all $\s{i}$. 
For any $T > 0$, since $A_{\s{i}}$ in \eqref{eq: LTV dyn} is continuous over $[0,T]$, each term in the Peano-Baker series\footnote{Though omitted here, the Peano-Baker series representation of the state transition matrix takes the form as shown in \cite[eq.~(12) of Chapter~3]{rugh1996linear}.} associated with the state transition matrix $\Phi_{\s{i}}(t,\tau)$ given in \eqref{eq: Phi} is continuous over $\mathcal{D}_T$. Together with the fact that the Peano-Baker series converges absolutely and uniformly for $(t,\tau) \in \mathcal{D}_T$ \cite[Theorem~3.3]{rugh1996linear}, this implies that $\Phi_{\s{i}} (t,\tau)$ is (jointly) continuous over the domain $\mathcal{D}_T$ \cite[Theorem~7.12]{rudin1976principles}. %(the uniform limit theorem). 
Since the product of jointly continuous functions is jointly continuous, the result follows from the continuity of $A_{\s{i}}$, $B_{\s{i}}$, and~$C_{\s{i}}$ supposed in \eqref{eq: LTV dyn}, and the construction of the kernel in \eqref{eq: g intro}, for all $\s{i} \in \{\s{p}, \s{c}\}$. %construction of the kernels in \eqref{eq: g intro}.
%This implies that the state-transition matrix~$\Phi(t,\tau)$ as described in \eqref{eq: Phi} is continuous in both variables. 
\end{proof}
%Now, given \eqref{eq: plant g_p}, \eqref{eq: contr g_c}, and \eqref{eq: up gq}, the aim is to analyze the effect of the FDIA~$y_{\s{a}}$, as modeled in \eqref{eq: FDIA LTV}, on $u_{\s{q}}$. That is, in our model of $u_{\s{q}}$, we want to account for the possible contribution of $y_\s{a}$ for stealth analysis. To this end, 
Next, we substitute $u_{\s{c}} = y_{\s{p}} + y_{\s{a}}$ (see Fig.~\ref{fig: system}) in \eqref{eq: plant g_p}, for $\s{i} = \s{c}$ (where $z_{\s{c}} \equiv u_{\s{q}}$), and apply \eqref{eq: FDIA LTV} to obtain
%we expand \eqref{eq: plant g_p} (for $\s{i} = \s{c}$, where $z_{\s{c}} \equiv u_{\s{q}}$) by substituting $u_{\s{c}} = y_{\s{p}} + y_{\s{a}}$ (see Fig.~\ref{fig: system}), and using \eqref{eq: FDIA LTV}, to obtain
\begin{align}
     u_{\s{q}} (t) = & \,  \int_0^t g_{\s{c}}(t,\sigma) \bigl( y_\s{p}(\sigma) + y_{\s{a}}(\sigma) \bigr) \, \dd \sigma, \nonumber \\
     = & \, \int_0^t g_\s{c}(t,\sigma)\, \biggl(\int_0^\sigma g_\s{p}(\sigma,\tau) \, u_\s{p}(\tau) \, \dd \tau\biggr) \, \dd \sigma \nonumber \\
     &+  \int_0^t g_\s{c}(t,\sigma)\, \biggl(\int_0^\sigma g_{\s{a}}(\sigma-\tau) \tilde{h}(\tau) \, \dd \tau \biggr) \, \dd \sigma, \label{eq: pre transition} \\
     = & \, \int_0^t \, \underbrace{\biggl(\int _{\tau}^t g_\s{c}(t, \sigma) \, g_\s{p}(\sigma, \tau) \, \dd \sigma \biggr)}_{\eqqcolon G_{\s{c},\s{p}}(t, \tau)} \,u_\s{p}(\tau) \, \dd \tau \nonumber \\
     & \, +\int_0^t \, \underbrace{\biggl( \int_\tau ^t g_\s{c}(t,\sigma)\, g_{\s{a}}(\sigma-\tau) \, \dd \sigma \biggr)}_{\eqqcolon G_{\s{c},\s{a}}(t,\tau)} \, \tilde{h}(\tau) \, \dd \tau, \label{eq: gcp gca} %}_{\coloneqq \phi_{\s{c},\s{a}}(t)} 
\end{align}
where the transition from \eqref{eq: pre transition} to \eqref{eq: gcp gca} was made by iterating the integrals, which is allowed because of the continuity of the integrands  following Lemma~\ref{lem: g cont} %,whereby equivalence is maintained under the interchange of the order of integration 
\cite[Chapter~5, Theorem~7]{protter2012intermediate}. Furthermore, in \eqref{eq: gcp gca}, we observe that the contribution from the FDIA appears in the following term, which we simplify by employing the definition of $g_{\s{a}}$ and $G_{\s{c},\s{a}}$ from \eqref{eq: FDIA LTV} and \eqref{eq: gcp gca}, respectively, as follows:
\begin{align} \label{eq: phi}
\phi_{\s{c},\s{a}}(t) \coloneqq  \int_{0}^tG_{\s{c},\s{a}}(t,\tau)\tilde{h}(\tau) \dd \tau = \frac{h}{a!} \int_0^t g_{\s{c}}(t,\tau) \tau^a \dd \tau.
\end{align}
Then, by substituting \eqref{eq: up gq} and \eqref{eq: phi} into \eqref{eq: gcp gca}, and interchanging the order of integration again for the first term, we obtain the following representation of $ u_{\s{q}}$:
%\vspace{0.05cm}
%\begin{mdframed}    
\begin{equation} \label{eq: uq volt}
     u_\s{q}(t) = \int _0^t G_{\s{c},\s{p}, \s{q}} (t,\tau) \, u_{\s{q}}(\tau)\, \dd \tau + \phi_{\s{c}, \s{a}}(t),% \quad t \geq 0,
\end{equation}
%\end{mdframed}
where 
%\begin{equation} \label{eq: Gcpq}
$G_{\s{c},\s{p}, \s{q}} (t,\tau)\coloneqq \int_\tau^tG_{\s{c},\s{p}}(t,\sigma) \, g_{q}(\sigma - \tau) \, \dd \sigma,\, 0\leq \tau \leq t$.  
%\end{equation}
%\vspace{0.05cm}
%\begin{mdframed}
We note at once that \eqref{eq: uq volt} is an LVIE of the second kind. Furthermore, under Lemma~\ref{lem: g cont}, \eqref{eq: uq volt} conforms to the well-posed structure assumed for \eqref{eq: generic LVIE} in Section~\ref{sec: prelim LVIE}, thereby guaranteeing the existence and uniqueness of the solution~$u_{\s{q}}$ to \eqref{eq: uq volt} over $t \geq 0$, corresponding to the continuous input~$\phi_{\s{c}, \s{a}}$ defined in \eqref{eq: phi}. %, thereby permitting the existence and uniqueness of the solution~$u_{\s{q}}$ to \eqref{eq: uq volt} over $t \geq 0 $.

By modeling $u_{\s{q}}$ in the form \eqref{eq: uq volt}, we transform the problem of stealth for the FDIA to the problem of stability for \eqref{eq: uq volt}. In the next subsection, we employ this insight to address Problem~\ref{prob: prob}. 

\subsection{Stealth and Linear Volterra Integral Equations}
%Firstly, the well-posedness of \eqref{eq: uq volt} and \eqref{eq: uc volt} is established by 
%observing that
%assuming that for any $T \geq 0$, their respective kernels $G_{\s{c}, \s{p}, \s{q}}(t,s)$ and $G_{\s{c}}(t,s)$, $0 \leq s \leq t$, are continuous over the domain~$\mathcal{D}$, where the $\mathcal{D}$ is as defined under \eqref{eq: generic LVIE}.
In this subsection, we shall present the main results of the study as they pertain to Problem~\ref{prob: prob}. % Problem~\ref{prob: prob}. 
%To that end, first, we consider the deviation $u_{\s{q}}$: 
That is, given the LVIE \eqref{eq: uq volt}, we seek conditions under which, corresponding to the input~$\phi_{\s{c}, \s{a}}$ modeled as in \eqref{eq: phi}, its solution~$u_{\s{q}}$ exhibits the following properties: i)~for a given $\epsilon >0$, it remains $\epsilon \text{-bounded}$ for all time ($\epsilon \text{-stealth}$), and ii)~$\lim_{t \to \infty} u_{\s{q}}(t) = 0$ (untraceable stealth). %To this end, we relate these properties to the notions of stability presented in Definitions~\ref{def: Stability ZB} and~\ref{def: AS ZB} for LVIEs. 
To that end, we note that if the LVIE is stable per Definition~\ref{def: Stability ZB}, then, for any $\epsilon > 0$, there exists $\delta(\epsilon) >0$ such that if we bound the input~$\phi_{\s{c}, \s{a}}$ by $\delta(\epsilon)$, then $u_{\s{q}}$ will remain $\epsilon \text{-bounded}$ for all $t$. Similarly, if the LVIE \eqref{eq: uq volt} is asymptotically stable, then there would exist some other bound~$\delta_1$ such that if the input $\phi_{\s{c}, \s{a}}$ remains $\delta_1 \text{-bounded}$ and, additionally, the input converged to~$0$ in the asymptotic sense, then $\lim_{t \to \infty}u_{\s{q}}(t) = 0$. Therefore, the stability of the LVIE provides us with a framework to argue about the stealth of~$y_{\s{a}}$ with respect to $u_{\s{q}}$. However, stability alone is not enough. Additionally, we need the input~$\phi_{\s{c},\s{a}}$: i)~bounded for $\epsilon\text{-stealth}$ and ii)~asymptotically converging to~$0$ for untraceable stealth, as discussed above. 

In light of the above discussion, we present the following sufficient conditions to obtain the sought-after attributes for~$\phi_{\s{c}, \s{a}}$.
\begin{lem} \label{lem: FDIA and stealth}
Consider the polynomial FDIA signal $y_{\s{a}}$ of degree $a \in \mathbb{Z}_{\geq 0}$ and weight $h \in \mathbb{R}$, as modeled in \eqref{eq: FDIA LTV}. %, and recall that $a \in \mathbb{Z}_{\geq 0}$ denotes the degree of the polynomial and $h \in \mathbb{R}$ its input weight. %number of integrators in the attack channel and $h\in \mathbb{R}$ the weight of the step input, as described under \eqref{eq: FDIA LTV}. 
\begin{enumerate}[label=\alph*)]
    \item \label{lem: FDIA and stealth: epsilon} For any $\delta > 0$, if
    \begin{equation} \label{eq: gctaua bounded}
        \sup_{t\geq 0} \int_0^t g_c(t,\tau)\, \tau^a \dd \tau < \infty, 
    \end{equation}
    then there exists~$h \in \mathbb{R} \setminus 0$ such that the input~$\phi_{\s{c},\s{a}}$ defined as in \eqref{eq: phi} remains $\delta \text{-bounded}$:
    \begin{equation} \label{eq: phi bounded}
        \sup_{t\geq 0} |\phi_{\s{c},\s{a}}(t)| < \delta. 
    \end{equation}
\item  Furthermore, if 
\begin{equation} \label{eq: gctaua AS phi}
        \lim_{t \to \infty} \int_0^t g_\s{c}(t,\tau)\, \tau^a \dd \tau = 0, 
    \end{equation}
    then 
    \begin{equation}\label{eq: phi AS}
        \lim_{t \to \infty} \phi_{\s{c},\s{a}}(t) = 0. 
    \end{equation}
\end{enumerate}
\end{lem}
\begin{proof}
\begin{enumerate}[label=\alph*)]
    \item Let $M\coloneqq \sup_{t\geq 0} \int_0^t g_c(t,\tau)\, \tau^a \dd \tau$. 
From \eqref{eq: phi}, we obtain 
\begin{align}
    \sup_{t\geq 0}|\phi_{
    \s{c},\s{a}}(t)| \leq \frac{|h|}{a!} |M|. 
\end{align}
Choosing~$h \neq 0$ such that $|h| < \frac{a!\delta}{|M|}$ then ensures \eqref{eq: phi bounded}. 
\item It follows from the application of the limit `$\lim_{t \to \infty}$' to \eqref{eq: phi} under \eqref{eq: gctaua AS phi}. 
\end{enumerate}
This concludes the proof.
\end{proof}
Lemma~\ref{lem: FDIA and stealth} stipulates that if the impulse-response kernel of the controller $g_\s{c}$ has sufficient decay to ensure integrability against the polynomial growth term $\tau^{a}$, where $\tau \in [0, t]$, $t \geq 0$, and $a \in \mathbb{Z}_{\geq0}$,  
such that \eqref{eq: gctaua bounded} and \eqref{eq: gctaua AS phi} hold, 
then $\phi_{\s{c}, \s{a}}$ remains bounded and asymptotically converges to $0$, as stated in \eqref{eq: phi bounded} and \eqref{eq: phi AS}, respectively. Without having assumed a prior structure for $g_\s{c}$, Lemma~\ref{lem: FDIA and stealth} then dictates the mathematical (decay) properties needed for $g_{\s{c}}$ to bound~$\phi_{\s{c},\s{a}}$. 
%With regard to Problem~\ref{prob: prob}, first, we consider the $\epsilon \text{-boundedness}$ of $\phi_{\s{c},{a}}$ per \eqref{eq: gctaua bounded}--\eqref{eq: phi bounded}. The intuition here is that if the controller is designed to keep the closed-loop containing a chain of~$q$ integrators stable in the absence of an FDIA~$y_{\s{a}}$, then the controller is also expected to stabilize the contribution of the FDIA in its output $u_{\s{q}}$ (in the sense of \eqref{eq: phi bounded}) that is output from a chain of $a$ integrators, given $a \leq q$ (as in the LTI case). To make the argument rigorous, we present the following result. 
%where $a \in \mathbb{Z}_{\geq 0}$ denotes the number of integrators in the attack channel) such that it yields the bounded and convergent (in the limit) convolution-like integrals in \eqref{eq: gctaua bounded} and \eqref{eq: gctaua AS phi}, respectively, then $\phi_{\s{c}, \s{a}}$ remains bounded and asymptotically converges to $0$, respectively.
%To see when the controller from the considered setup would permit such a decay to ensure \eqref{eq: gctaua bounded}, we present the following result. 
%However, first, we impose the following assumption. 
To proceed with the strategy discussed before, the following assumption is warranted:
\begin{assum} \label{assm: uq < infty LVIE} We suppose the following: %with regards to \eqref{eq: uq volt}:
\begin{enumerate}[label=\alph*)]
        \item The LVIE \eqref{eq: uq volt} is stable per Definition~\ref{def: Stability ZB}. %\footnote{To justify this assumption, we posit that if the converse were true, the current study would be pointless. The implications of this assumption forms one of the main contributions of this section.}.
        %\item The kernel $$G_{c,p,\bar{p}}(t,s) \geq 0, \quad 0 \leq s \leq t < \infty.$$ 
        \item $\sup_{t \geq 0}$ $\int_0^1 g_c(t,\tau) \, \dd \tau < \infty$.
        \item For any fixed~$t \geq0$,~$$G_{\s{c},\s{p}}(t, \tau) \geq g_\s{c}(t, \tau) \geq 0, \quad \forall \tau : 0 \leq \tau\leq t,$$
        where the kernel $G_{\s{c},\s{p}}$ is as defined in \eqref{eq: gcp gca}.
    \end{enumerate}
\end{assum}
\begin{rem}
Assumption~\ref{assm: uq < infty LVIE}-(a) ensures the internal stability of the closed-loop mapping of the controller's initial states to its output, %system so that the with respect to the output of the controller, 
which is essential for the well-posedness of Problem~\ref{prob: prob}. 
%ensures the input--output stability of the `system' that maps $\tilde{h}$ (see \eqref{eq: FDIA LTV}) to the output of the controller, $u_{\s{q}}$, which is fundamental to our ensuing analysis. %, which is needed for us to proceed with the stealth analysis. 
Assumption~\ref{assm: uq < infty LVIE}-(b) further guarantees the uniform boundedness of the controller's kernel, which is required for its (BIBO) stability. %input--output behavior. %standard for a stable controller.
%thereby allowing for internal stability.  %then guarantees that the impulse response kernel of the controller remains bounded for all time, which, again, is required for our stability analysis. %$\tau \text{-integral}$ of $g_\s{c}(t,\tau)$ remains bounded for all time, which is needed for \eqref{eq: }is needed to keep the associated Volterra kernel $G_{\s{c},\s{p}, \s{q}}$ from becoming unbounded.
%Alternatively, a structural constraint on the kernels $g_{\s{c}}$ and $g_{\s{p}}$ could be imposed such that they then yield these conditions. However, we do not impose such restrictions in this study. 
\mbox{Assumption~\ref{assm: uq < infty LVIE}-(c)}, however, is admittedly restrictive. Still, we adopt it in this study because i)~it ensures that the kernel~$G_{\s{c},\s{p}, \s{q}}(t,\tau)$ of \eqref{eq: uq volt} is non-negative over the domain of interest, thereby permitting the direct application of Lemmas~\ref{lem: iff stable ZB} and~\ref{lem: iff AS ZB}, and ii)~it allows us to focus the analysis on $g_{\s{c}}(t,\tau)$ instead of~$G_{\s{c},\s{p}}(t,\tau)$ without having to impose further structural constraints on $g_{\s{p}}(t,\tau)$ to state our results. Relaxing this condition %in the context of Problem~\ref{prob: prob} 
is a topic for future work. %ensure the LVIE \eqref{eq: uq volt} under study is stable and the controller bounded, thereby allowing further analysis given our LVIE-based analysis approach. 
%Without assuming the stability of \eqref{eq: uq volt} and the boundedness of the controller kernel, respectively, we cannot proceed.  
\end{rem}
Having established the needed assumptions, we can present the following result.
%Returning to the problem at hand, we state the following result. 
\begin{lem} \label{lem: gc tauq < infty}
Under Assumption~\ref{assm: uq < infty LVIE}, the following are true:
    \begin{enumerate}[label=\alph*)]
    \item For $q \in \mathbb{Z}_{\geq 1}$, %denotes the number of integrators in $\mathcal{Q}$, as described in \eqref{eq: up gq}, 
    \begin{equation} \label{eq: gctauq infty}
        \sup_{t \geq 0} \,  \int_0^t g_\s{c}(t, \tau) \, \tau^q \, \dd \tau \, < \infty.
    \end{equation}
    \item For $q \in \mathbb{Z}_{\geq 1}$ and $a \in \mathbb{Z}_{\geq 0}$, if $a \leq q$, then \eqref{eq: gctaua bounded} holds.
%    \begin{equation} \label{eq: gctaua infty}
%        \sup_{t \geq 0} \,  \int_0^t g_\s{c}(t, \tau) \, \tau^a \, \dd \tau \, < \infty.    \end{equation}    
    \end{enumerate}
\end{lem}
A proof for Lemma~\ref{lem: gc tauq < infty} is included in Appendix~\ref{sec: appndx lemma gc tauq < infty}. %As expected, the result illustrates that the stability of \eqref{eq: uq volt} implies the controller can tolerate the contribution of an FDIA designed as the output of $a$ integrators, given $a \leq q$. 
Lemma~\ref{lem: gc tauq < infty} dictates that, under Assumptions~\ref{assm: uq < infty LVIE}-(b) and~\ref{assm: uq < infty LVIE}-(c), the stability of the LVIE \eqref{eq: uq volt} supposed in Assumption~\ref{assm: uq < infty LVIE}-(a) implies that the controller's kernel $g_\s{c}$ has sufficient decay to maintain the integrability of the convolution-like integral in \eqref{eq: gctauq infty} despite the polynomial growth term~$\tau^q$, where $\tau \in [0,t]$, $t \geq 0$, and $q\in \mathbb{Z}_{\geq 1}$. Consequently, if $a \leq q$, then the controller can also handle the growth term $\tau^a$ in \eqref{eq: gctaua bounded}, for $\tau \in [0,t]$, $t \geq 0$, and $a \in \mathbb{Z}_{\geq 0}$. %, as per the integral \eqref{eq: gctaua bounded}. 
This brings us to the first main result of the section. 
\begin{thrm} \label{thm: bounded attack}
Suppose Assumption~\ref{assm: uq < infty LVIE} holds. %, and recall that $a \in \mathbb{Z}^+$ denotes the number of integrators in the attack channel and $h\in \mathbb{R}$ the weight of the step input, as described under \eqref{eq: FDIA LTV}. 
For $q \in \mathbb{Z}_{\geq 1}$ and $a \in \mathbb{Z}_{\geq 0}$ described as in \eqref{eq: gq} and \eqref{eq: FDIA LTV}, respectively, if~$a \leq q$, then, for any given $\epsilon > 0$, there exists~$h \in \mathbb{R} \setminus 0$ such that the corresponding FDIA $y_{\s{a}}$ designed per \eqref{eq: FDIA LTV} remains $\epsilon\text{-stealthy}$ with respect to $ u_{\s{q}}$ per Definition~\ref{def: stealth}.
\end{thrm}

\begin{proof}
Under Assumption~\ref{assm: uq < infty LVIE}-(a), for any given $\epsilon > 0$, $\exists \delta(\epsilon) >0$ such that $\sup_{t \geq 0} | \phi_{\s{c}, \s{a}}(t)| < \delta(\epsilon)  \implies \sup_{t\geq 0} |u_{\s{q}}(t)| < \epsilon$ (see Definition~\ref{def: Stability ZB}). Therefore, if $h$ can be chosen such that $\sup_{t\geq 0} |\phi_{\s{c},\s{a}}(t)| < \delta(\epsilon)$ then $\sup_{t \geq 0} |u_{\s{q}}(t)| < \epsilon$, which implies $\epsilon \text{-stealth}$ for the FDIA $y_{\s{a}}$ with respect to $u_{\s{q}}$ per Definition~\ref{def: stealth}. This is established through Lemmas~\ref{lem: FDIA and stealth}-(a) %(see \eqref{eq: gctaua bounded}--\eqref{eq: phi bounded}) 
and~\ref{lem: gc tauq < infty}-(b), which completes the proof.
%We have shown that $(a \leq q) \implies\sup_{t \geq 0} |\phi_{\s{c}, \s{a}(t)}| < \infty$. Equivalently, this means that there exists a finite constant $M > 0$ such that $\sup_{t \geq 0} |\phi_{\s{c}, \s{a}}(t)| = M < \infty$. Under Assumption~\ref{assm: uq < infty LVIE}, since  \eqref{eq: uq volt} is stable per Definition~\ref{def: Stability ZB}, there exists $\delta(\alpha)$ such that $\sup_{t\geq 0} |\phi_{\s{c}, \s{a}}(t)| < \delta (\alpha) \implies \sup_{t \geq 0} |\delta u_{\s{q}}(t)| < \alpha$. By setting $|h| = \frac{a!\delta (\alpha)}{M}$ in \eqref{eq: proof |a| < infty}, we attain stealth for the FDIA and conclude the proof. 
\end{proof}

We move on to pursue the conditions enabling untraceable stealth for the FDIA with respect to $u_{\s{q}}$. Similar to $\epsilon \text{-stealth}$, for untraceable stealth, we need to ensure appropriate behavior of the input $\phi_{\s{c}, \s{a}}$ to \eqref{eq: uq volt} such that u.s. follows from the stability properties stated in Definition~\ref{def: AS ZB}. Keeping this in mind, first, we make the following assumption. 
\begin{assum} \label{assm: uq AS LVIE} We suppose:
\begin{enumerate}[label=\alph*)]
    \item The LVIE in \eqref{eq: uq volt} is asymptotically stable per Definition~\ref{def: AS ZB}.
    \item $\lim_{t \to \infty}\int_{0}^1g(t,\tau)\dd\tau = 0$.
    \item $\lim_{T \to 0^+ } \sup_{t \geq 0} \int _0^T g_{\s{c}}(t,\tau) \dd\tau = 0$.
\end{enumerate}
\end{assum}
\begin{rem} %Assumptions~\ref{assm: uq AS LVIE}~(a)-(b) are standard: 
As discussed earlier, to prove u.s. for the FDIA \eqref{eq: FDIA LTV}, we rely on \eqref{eq: uq volt} being asymptotically stable, which then necessitates Assumption~\ref{assm: uq AS LVIE}-(a). Assumptions~\ref{assm: uq AS LVIE}-(b) and~\ref{assm: uq AS LVIE}-(c) are additionally introduced to impose a regularity constraint on the asymptotic behavior of the controller's impulse-response, and ensure the continuity of $g_{\s{c}}(t,\tau)$ near $\tau = 0$, respectively, which we shall require to present our results.
\end{rem}
With Assumption~\ref{assm: uq AS LVIE} in place, we state the following result. 
%This leads us to the following result.
\begin{lem} \label{lemm: AS}
    Under Assumptions~\ref{assm: uq < infty LVIE} and \ref{assm: uq AS LVIE}, the following are true:
    \begin{enumerate}[label=\alph*)]
    \item For $q \in \mathbb{Z}_{\geq 1}$, % denotes the number of integrators in $\mathcal{Q}$, as described in \eqref{eq: up gq}, 
    \begin{equation} \label{eq: gctauq-1 AS}
        \lim_{t \to \infty} \int_0^t g_\s{c}(t,\tau) \, \tau^{q-1} \, \dd \tau = 0.
    \end{equation}
    \item For $q \in \mathbb{Z}_{ \geq 1}$ and $a \in \mathbb{Z}_{\geq 0}$, if $a < q$, then \eqref{eq: gctaua AS phi} holds.
%    \begin{equation} \label{eq: gctaua AS}
%         \lim_{t \to \infty}  \int_0^t g_\s{c}(t, \tau) \, \tau^a \, \dd \tau \, = 0.    \end{equation}    
    \end{enumerate}
\end{lem}

This leads us to the second main result of the section.
%Now we can state the second main result of the section.
\begin{thrm} \label{thm: u.s.}
Suppose Assumptions~\ref{assm: uq < infty LVIE} and \ref{assm: uq AS LVIE} hold. For $q \in \mathbb{Z}_{\geq 1}$ and $a \in \mathbb{Z}_{\geq 0}$ described as in \eqref{eq: gq} and \eqref{eq: FDIA LTV}, respectively, if $a < q$, then there exists $h \in \mathbb{R} \setminus 0$ such that the corresponding FDIA~$y_a$ designed per \eqref{eq: FDIA LTV} remains untraceably stealthy with respect to $u_{\s{q}}$ per Definition~\ref{def: stealth}. 
\end{thrm}

\begin{proof}
Firstly, the FDIA $y_{\s{a}}$ is deemed $\epsilon \text{-stealthy}$ with respect to $u_{\s{q}}$ per Theorem~\ref{thm: bounded attack}.
Then, since \eqref{eq: uq volt} is assumed to be asymptotically stable, we infer (from Definition~\ref{def: AS ZB}) that there exists $\delta >0$ such that $\sup_{t\geq 0}| \phi_{\s{c}, \s{a}}(t) | < \delta$ and $\lim_{t \to \infty} \phi_{\s{c}, \s{a}}(t) = 0 \implies \lim_{t \to \infty} u_{\s{q}}(t) = 0$, which then renders the FDIA~$y_{\s{a}}$ u.s. with respect to $u_{\s{q}}$ per Definition~\ref{def: stealth}. Therefore, to complete the proof, it suffices to show that there exists $h \neq 0$ such that i)~$\sup_{t\geq 0}|u_{\s{q}}(t)|< \delta$ and ii)~$\lim_{t \to \infty} u_{\s{q}}(t) = 0$. The first condition follows from Lemmas~\ref{lem: FDIA and stealth}-(a) and~\ref{lem: gc tauq < infty}-(b), and the second from Lemmas~\ref{lem: FDIA and stealth}-(b) and~\ref{lemm: AS}-(b). This completes the proof. 
%Under Assumption~\ref{assm: uq AS LVIE}, it follows from Definition~\ref{def: AS ZB} that there exists some~$\delta >0$ such that if $\forall t, \, |\phi_{\s{c}, \s{a}}(t)|  < \delta$ and  $\lim_{t \to \infty}\phi_{\s{c}, \s{a}}(t)=0$ then $\lim_{t \to \infty}u_q(t) = 0$.
%From \eqref{eq: phi}, we note that if $a = q-1$, then 
%    \begin{equation}
%      \lim_{t \to \infty}|\phi_{\s{c}, \s{a}}(t)| \leq \frac{|h|}{a!} \lim_{t \to \infty} \int_0^t g_{\s{c}}(t,\tau) \tau^{q-1} \dd \tau = 0,
%    \end{equation}
%    thereby implying $\lim_{t \to \infty}|\phi_{\s{c}, \s{a}}(t)| = 0$.
\end{proof}
%%Given that there are $q \in \mathbb{Z}_{\geq 1}$ integrators in a closed-loop with a positive-kernel feedback, Theorems~\ref{thm: bounded attack} and~\ref{thm: u.s.} stipulate that FDIAs of the polynomial form \eqref{eq: FDIA LTV} will remain stealthy if %(with respect to the polynomial model in \eqref{eq: FDIA LTV}) 
%their degree $a \leq q$. This
 %, given that there are $q \in \mathbb{Z}_{\geq 1}$ integrators in the closed-loop. %,  with respect to $u_{\s{q}}$ for a system with $q \in \mathbb{Z}_{\geq 1}$ integrators. 

%\begin{rem} \label{rem: unbounded} 
%Theorems~\ref{thm: bounded attack} and~\ref{thm: u.s.} provide us with a framework to characterize the type of FDIAs that will remain stealthy against integrator-endowed systems. % the considered setup in terms of their polynomial degree $a\in \mathbb{Z}_{\geq 0}$. 
%To put things into perspective, even a single integrator ($q = 1$) in such a closed-loop (that satisfies the required conditions stated in Theorems~\ref{thm: bounded attack} and~\ref{thm: u.s.}) will permit a polynomial attack signal of degree $a = 1$ (growing unbounded over time) in the sensor channel to remain $\epsilon \text{-bounded}$ with respect to $u_{\s{q}}$, for some $\epsilon > 0$. For u.s. of an unbounded injection, however, we require at least $q = 2$ integrators in the loop. 
%\end{rem}

\section{Numerical Examples} \label{sec: num ex}
In this section, we provide numerical validation of our theoretical results. 
\subsection*{Example~1 (Non-negative-kernel feedback):}
For this example, we consider the following instantiation of the control system depicted in Fig.~\ref{fig: system}: 
\begin{itemize}
\item The system $\mathcal{Q}$ is composed of a chain of $q =2$ integrators, and its impulse-response kernel $g_{\s{q}}(t) = t$, for $t \geq 0$.
\item The controller $\mathcal{C}$ is represented by the LTV system
\begin{equation}
\begin{split}    
    \dot{x}_{\s{c}} & = -t^2 x_{\s{c}}+ u_{\s{c}},\\ 
    u_{\s{q}} &= x_{\s{c}}.
\end{split}
\end{equation}
Correspondingly, its impulse-response kernel is given by
\begin{equation} \label{eq: gc Ex1}
    g_\s{c}(t,\tau) = e^{-\frac{1}{3}(t^3 - \tau^3)}, \quad 0 \leq \tau \leq t.
\end{equation} 
\item The plant~$\mathcal{P}$ is the trivial unity-gain mapping such that $y_{\s{p}}(t) \equiv u_{\s{p}}(t)$. Correspondingly, $g_{\s{p}}(t,\tau) \equiv \delta_\s{D}(t - \tau)$, where $\delta_\s{D}$ denotes the Dirac delta function. 
%\item The plant~$\mathcal{P}$ is the trivial unity-gain system such that $y_{\s{p}}(t) \equiv u_{\s{p}}(t)$. Its impulse response is then given by $g_{\s{p}}(t,\tau) \equiv \delta_\s{D}(t - \tau)$, where $\delta_\s{D}$ denotes the Dirac delta function.
%\item The integrator system~$\mathcal{Q}$ is composed of a chain of two single-integrators with dynamics as given in \eqref{eq: SISO integrator} for $q =2$, and its impulse-response kernel $g_{\s{q}}(t) = t$, $t \geq 0$.
%\item The controller~$\mathcal{C}$ is represented by the LTV system 
%\begin{equation}
%\begin{split}    
%    \dot{x}_{\s{c}} & = -t^2 x_{\s{c}}+ u_{\s{c}},\\ 
%    u_{\s{q}} &= x_{\s{c}},
%\end{split}
%\end{equation}
%with the impulse-response kernel given by
%\begin{equation} \label{eq: gc Ex1}
%    g_\s{c}(t,\tau) = e^{-\frac{1}{3}(t^3 - \tau^3)}, \quad 0 \leq \tau \leq t.
%\end{equation} 
\end{itemize}
%plant~$\mathcal{P}$ to be the trivial unity-gain system such that $y_{\s{p}}(t) \equiv u_{\s{p}}(t)$. Its impulse response can then be modeled as $g_{\s{p}}(t,\tau) \equiv \delta_\s{D}(t - \tau)$, where $\delta_\s{D}$ denotes the Dirac delta function. Furthermore, the integrator system~$\mathcal{Q}$ is given by a chain of two single-integrators with dynamics as given in \eqref{eq: SISO integrator} for $q =2$, and its impulse-response kernel $g_{\s{q}}(t) = t$, $t \geq 0$. %That is, it is a chain of two single-integrators. 
%Finally, the controller~$\mathcal{C}$ is the following LTV system 
%\begin{equation}
%\begin{split}    
%    \dot{x}_{\s{c}} & = -t^2 x_{\s{c}}+ u_{\s{c}},\\ 
%    u_{\s{q}} &= x_{\s{c}},
%\end{split}
%\end{equation}
%with the impulse-response kernel given by
%\begin{equation} \label{eq: gc Ex1}
%    g_\s{c}(t,\tau) = e^{-\frac{1}{3}(t^3 - \tau^3)}, \quad 0 \leq \tau \leq t.
%\end{equation} 
It can be verified that Assumptions~\ref{assm: uq < infty LVIE} and~\ref{assm: uq AS LVIE} hold for this setup. Therefore, Theorems~\ref{thm: bounded attack} and~\ref{thm: u.s.} apply. To validate these results, numerical simulations are included in Fig.~\ref{fig: Ex1}.

Firstly, we consider an FDIA constructed per \eqref{eq: FDIA LTV} for $a = 2$ and $h =1$, as shown in Fig.~\ref{fig: FDIA parabola}. For this case, since $a = q$, Theorem~\ref{thm: bounded attack} stipulates that $u_{\s{q}}$ will remain bounded for all $t$ by some finite $\epsilon > 0$. This is indeed confirmed in Fig.~\ref{fig: uq a 2}, which renders the FDIA growing unbounded over time $\epsilon \text{-stealthy}$, where $\epsilon < 0.4$. In a similar vein,  an FDIA constructed per \eqref{eq: FDIA LTV} for $a = 1$ and $h =1$ is then considered in Fig.~\ref{fig: FDIA ramp}. In this scenario, since $a < q$, Theorem~\ref{thm: u.s.} dictates that $u_{\s{q}}$ will not only remain bounded, but will also converge to $0$ asymptotically. The plot in Fig.~\ref{fig: uq a 1} confirms this, thereby rendering the (unbounded) FDIA untraceably stealthy with respect to $u_{\s{q}}$ per Definition~\ref{def: stealth}.
%Similarly, for an FDIA constructed per \eqref{eq: FDIA LTV} for $a = 1$ and $h =1$, as plotted in Fig.~\ref{fig: FDIA ramp}, and the response it elicits in the signal $u_{\s{q}}$ in Figs.~\ref{fig: FDIA ramp} and~\ref{fig: uq a 1}, respectively. In this scenario, since $a < q$, Theorem~\ref{thm: u.s.} dictates that the corresponding deviation in $u_{\s{q}}$ will not only remain $\epsilon \text{-bounded}$, for some $\epsilon > 0$, it will also converge to $0$ asymptotically. This is depicted in Fig.~\ref{fig: uq a 1}, thereby confirming that the FDIA in Fig.~\ref{fig: FDIA ramp} is indeed untraceably stealthy with respect to $u_{\s{q}}$ per Definition~\ref{def: stealth}. 
\begin{figure}[t]
    \centering
    % Top row: subfigure 1 and 2 side-by-side
    \begin{subfigure}[t]{0.24\textwidth}
        \centering
        % scale to width; if too tall, use height=<...> instead
        \includegraphics[width=\textwidth,keepaspectratio]{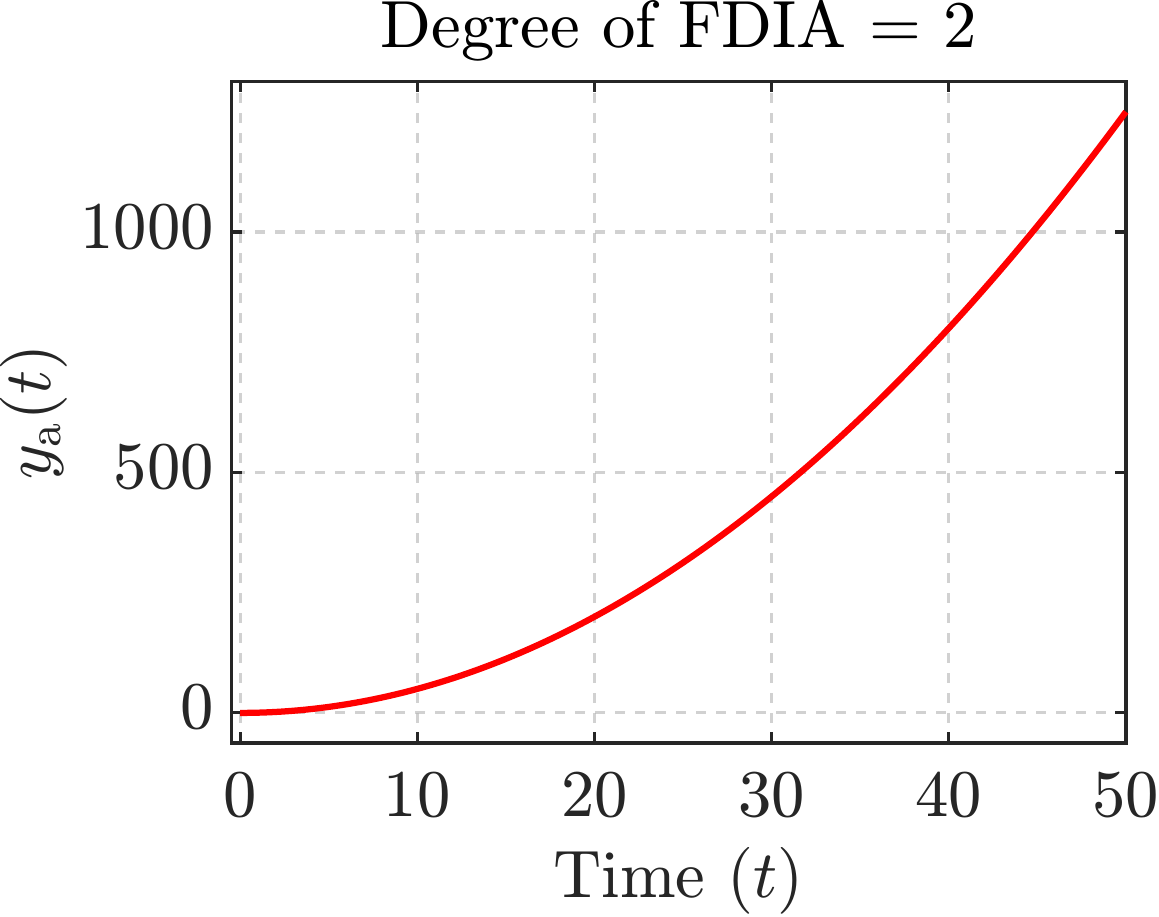}
        \caption{}
        \label{fig: FDIA parabola}
    \end{subfigure}%
    \hfill%
    \begin{subfigure}[t]{0.24\textwidth}
        \centering
        \includegraphics[width=\textwidth,keepaspectratio]{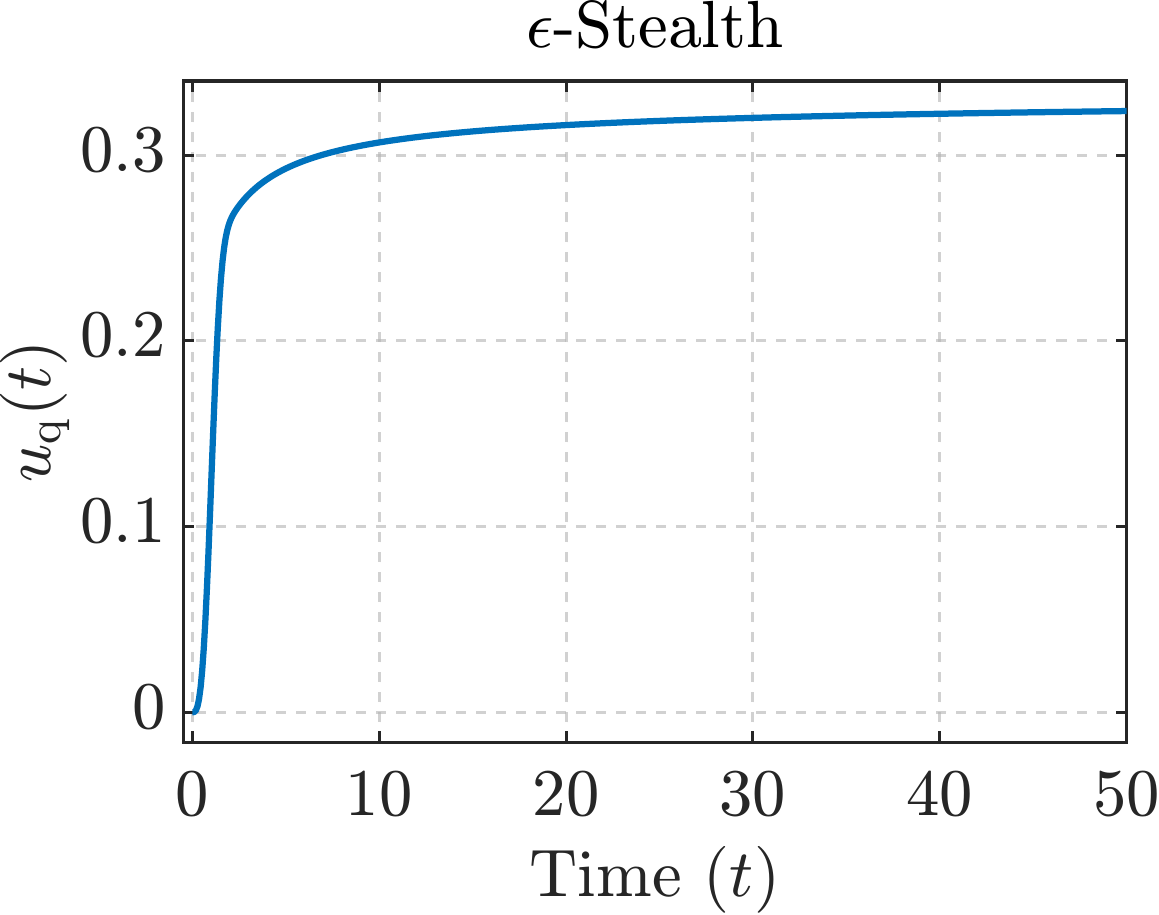}
        \caption{}
        \label{fig: uq a 2}
    \end{subfigure}
    % small vertical gap between rows (optional)
    \vspace{1pt}

    % Bottom row: subfigure 3 spanning almost full width
    \begin{subfigure}[t]{0.24\textwidth}
        \centering
        \includegraphics[width=\textwidth,keepaspectratio]{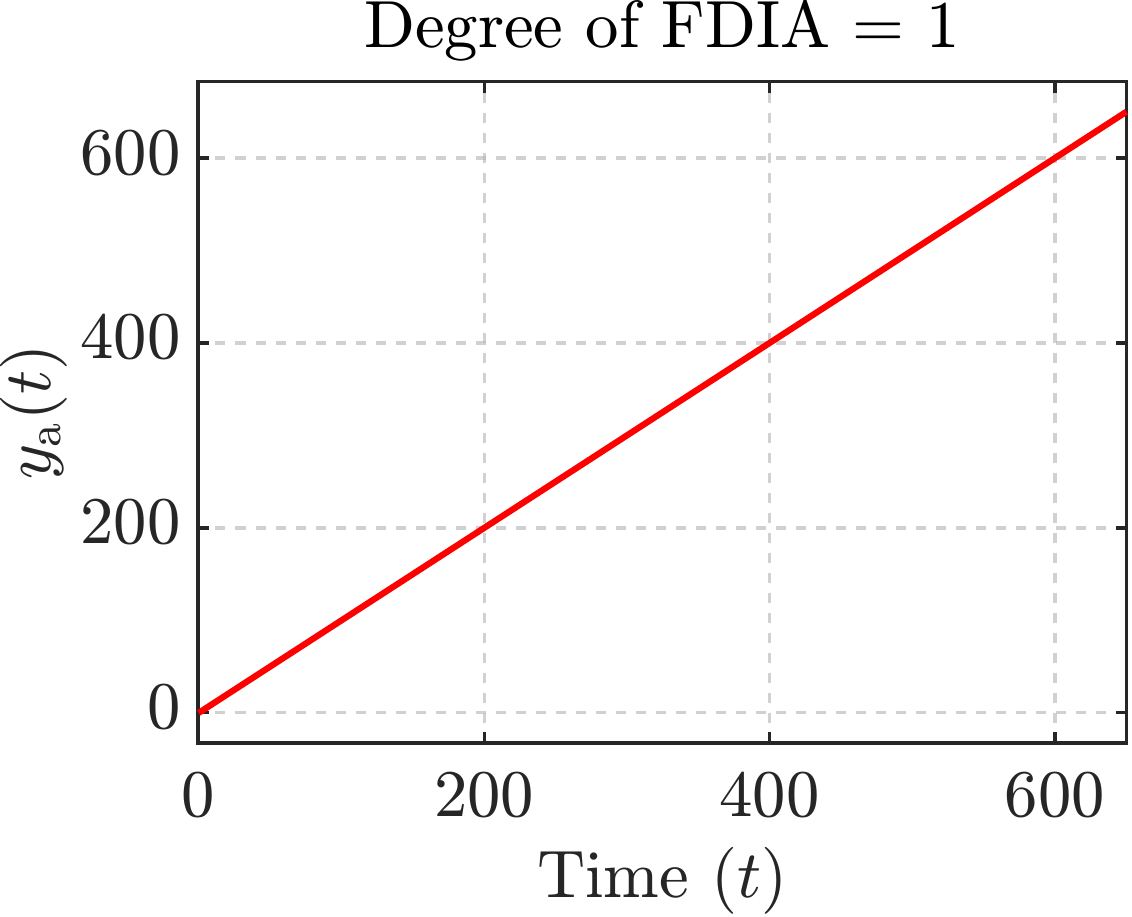}
        \caption{}
        \label{fig: FDIA ramp}
    \end{subfigure}
    \hfill%
    \begin{subfigure}[t]{0.24\textwidth}
        \centering
        \includegraphics[width=\textwidth,keepaspectratio]{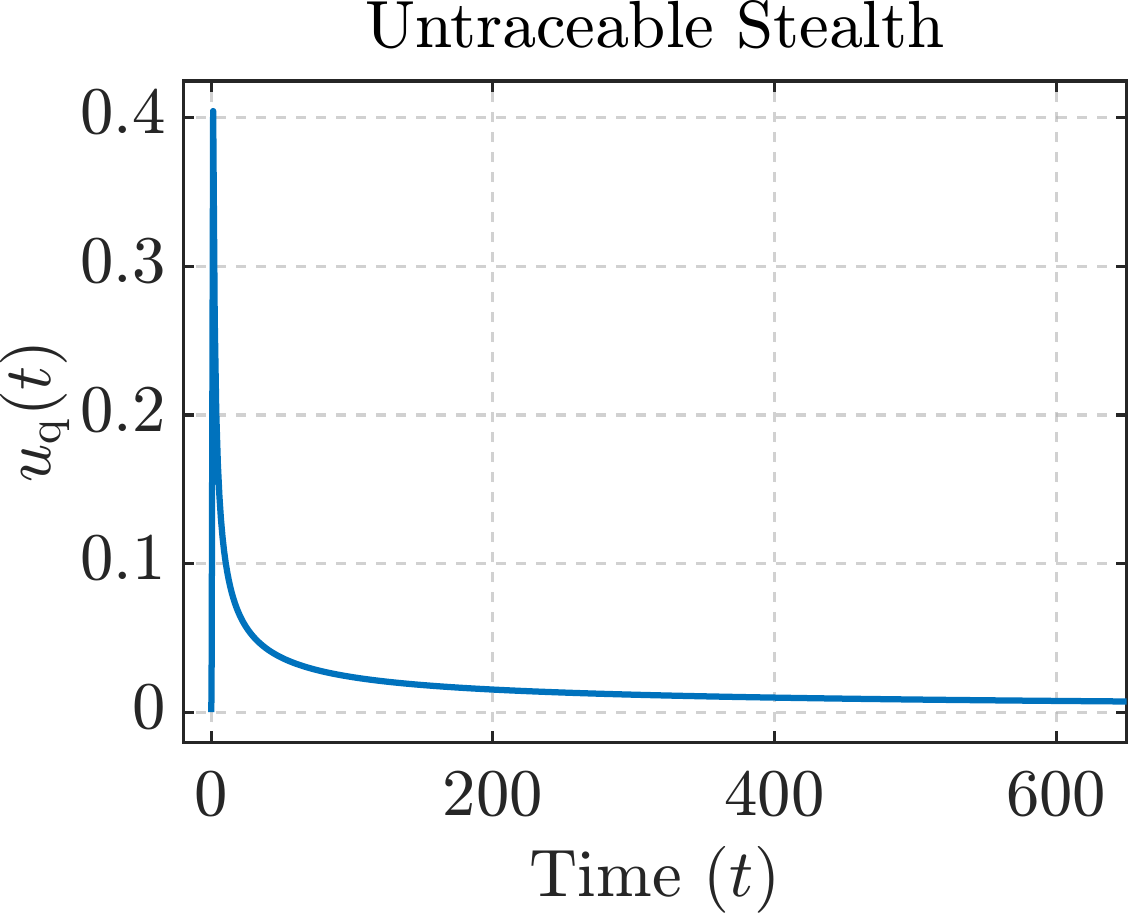}
        \caption{}
        \label{fig: uq a 1}
    \end{subfigure}
    %% small vertical gap between rows (optional)
    %\vspace{1pt}
    %% Bottom row: subfigure 3 spanning almost full width
    %\begin{subfigure}[t]{0.24\textwidth}
    %    \centering
    %    \includegraphics[width=\textwidth,keepaspectratio]{Figures/LTI/LTI case 3.pdf}
    %    \caption{}
    %    \label{fig: LTI case 3}
    %\end{subfigure}
    %\hfill%
    %\begin{subfigure}[t]{0.24\textwidth}
    %    \centering
    %    \includegraphics[width=\textwidth,keepaspectratio]{Figures/LTI/LTI attack cubic.pdf}
    %    \caption{}
    %    \label{fig: LTI attack cubic}
    %\end{subfigure}
    %% small vertical gap between rows (optional)
    %\vspace{1pt}
    \caption{Simulation results for Example~1 from Section~\ref{sec: num ex}.} %showing the propagation of the residual signals corresponding to the FDIA signals that grow unbounded over time.} 
    \label{fig: Ex1}
\end{figure}
% Optionally force figures to appear before continuing
%\FloatBarrier

%In light of the above discussion, Theorems~\ref{thm: bounded attack} and~\ref{thm: u.s.} stipulate that if the FDIA $y_{\s{a}}$ is constructed according to \eqref{eq: FDIA LTV} for $a \leq q = 2$, then it will remain $\epsilon\text{-bounded}$ (for some $\epsilon < 0$, depending on the choice of $h \in \mathbb{R}$ in \eqref{eq: FDIA LTV}) for all time. If $a$ is strictly less than $q$ ($a<q$), then, additionally, the FDIA will converge to $0$ asymptotically, thus achieving stealth in both the $\epsilon$ and the untraceable sense of the notion, per Definition~\ref{def: stealth}. We validate these claims through the simulation results illustrated in Fig.~\ref{fig: Ex1}.

%In Fig.~\ref{fig: uq a 2}, $u_{\s{q}}(t)$ is shown to vary over time~$t$ under the parabolic FDIA~$y_{a}$ constructed from \eqref{eq: FDIA LTV} with $a = 2$. As expected per Theorem~\ref{thm: bounded attack}, since $a = q$, $u_{\s{q}}$ remains $\epsilon \text{-bounded}$ for all~$t$, for $\epsilon \approx0.5$, despite the FDIA growing unbounded over time, as shown in Fig.~\ref{fig: FDIA parabola}. This renders the FDIA $\epsilon \text{-stealthy}$ with respect to $u_{\s{q}}$ per Definition~\ref{def: stealth}. Furthermore, in response to the ramp-type FDIA constructed under \eqref{eq: FDIA LTV} for $a = 1< q =2$, 

%The simulations presented in this section corroborate the theoretical results from Theorems~\ref{thm: bounded attack} and~\ref{thm: u.s.} implying 
The theoretical results and simulations presented so far corroborate the claim that integrators in the control-loop make \emph{non-negative}-kernel feedback systems susceptible to stealthy FDIAs. However, this does not imply that integrator-endowed systems under \emph{non-positive}-kernel feedback are exempt from suffering the same fate. To highlight this, next, we consider non-positive-kernel feedback in the context of Problem~\ref{prob: prob}.

\subsection*{Example~2 (Non-positive-kernel feedback):}
For this example, we again consider the control system adopted in Example~1, but with the following variation: The controller~$\mathcal{C}$ is now given by the following LTV system
\begin{equation}
    \begin{split}
        \dot{x}_{\s{c}} &= -(3t^2 + 0.5)x_{\s{c}} + u_{\s{c}}, \\
        u_{\s{q}} &= -x_c, 
    \end{split}
\end{equation}
which permits the non-positive impulse-response kernel %given by
\begin{equation}
    g_{\s{c}}(t,\tau) = -e^{-(t^3 -\tau^3) - \frac{1}{2}(t - \tau)}, \quad 0 \leq \tau < t.
\end{equation}
The response of this setup to the FDIA~$y_{\s{a}}$ constructed per \eqref{eq: FDIA LTV} for $a = 1$ and $h = 0.1$ is shown in Fig.~\ref{fig: Ex2}. 

As could be observed, this FDIA --- growing unbounded over time (see Fig.~\ref{fig: FDIA parabola}) --- causes the output of the system to grow unbounded over time, as shown in Fig.~\ref{fig: Ex2 yp}. However, the signals~$u_{\s{q}}$ and $u_{\s{c}}$ paint a different picture: The signal $u_{\s{q}}$ remains bounded for all $t$, thereby rendering the FDIA $\epsilon \text{-stealthy}$ with respect to $u_{\s{q}}$, where $\epsilon < 3$, as depicted in Fig.~\ref{fig: Ex2 uq a 1}. Similarly, the signal $u_{\s{c}}$ remains bounded for all $t$ and, additionally, converges to $0$ asymptotically. This then renders the FDIA u.s. with respect to $u_{\s{c}}$, as could be seen in Fig.~\ref{fig: Ex2 uc a 1}. 

%It can be observed that, corresponding to the FDIA that grows linearly unbounded over time, the output of the system also grows unbounded but in the opposite (sign) direction. However, as shown in Fig.~\ref{fig: Ex2 uq a 1}, the deviation produced in $u_{\s{q}}$ in response to this FDIA remains bounded and, additionally, converges to $0$ in the asymptotic sense (untraceable stealth). %This demonstrates untraceable stealth for the FDIA with respect to $u_{\s{q}}$ per Definition~\ref{def: stealth}. 
%Additionally, Fig.~\ref{fig: Ex2 uc a 1} shows that $u_{\s{c}}$ (the input to the controller) also remains bounded for all $t$ in response to the linearly unbounded FDIA, thereby rendering the FDIA 
%$\epsilon \text{-stealthy}$ (for some $\epsilon > 0$) with respect to $u_{\s{c}}$ as well. 
%. Therefore, we can conclude that the FDIA exhibits $\epsilon \text{-stealth}$ (for some $\epsilon > 0$) with respect to~$u_{\s{c}}$ and $u_{\s{q}}$ both and, additionally it demonstrates untraceable stealth with respect to $u_{\s{q}}$. 

\begin{figure}[t]
    \centering
    % Top row: subfigure 1 and 2 side-by-side
    \begin{subfigure}[t]{0.24\textwidth}
        \centering
        % scale to width; if too tall, use height=<...> instead
        \includegraphics[width=\textwidth,keepaspectratio]{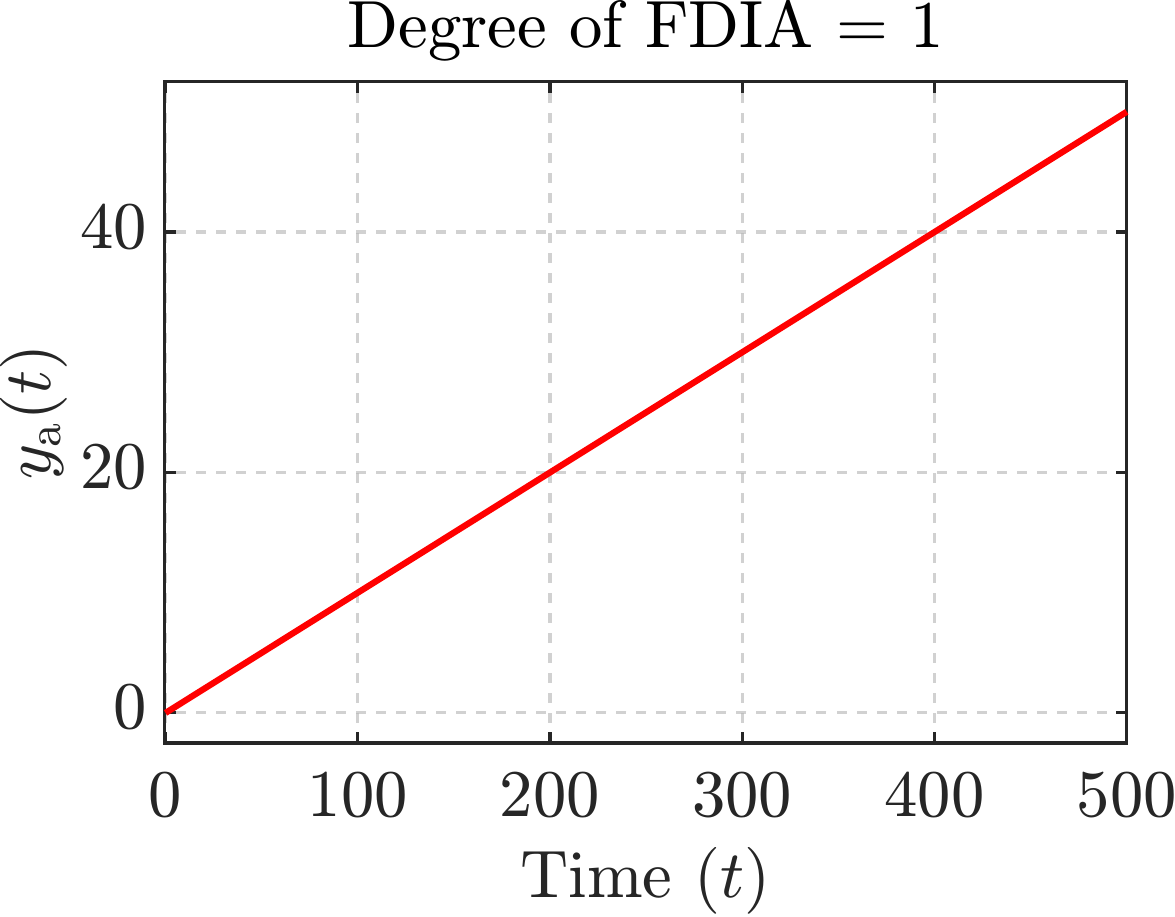}
        \caption{}
        \label{fig: FDIA Ex2}
    \end{subfigure}%
    \hfill%
    \begin{subfigure}[t]{0.24\textwidth}
        \centering
        \includegraphics[width=\textwidth,keepaspectratio]{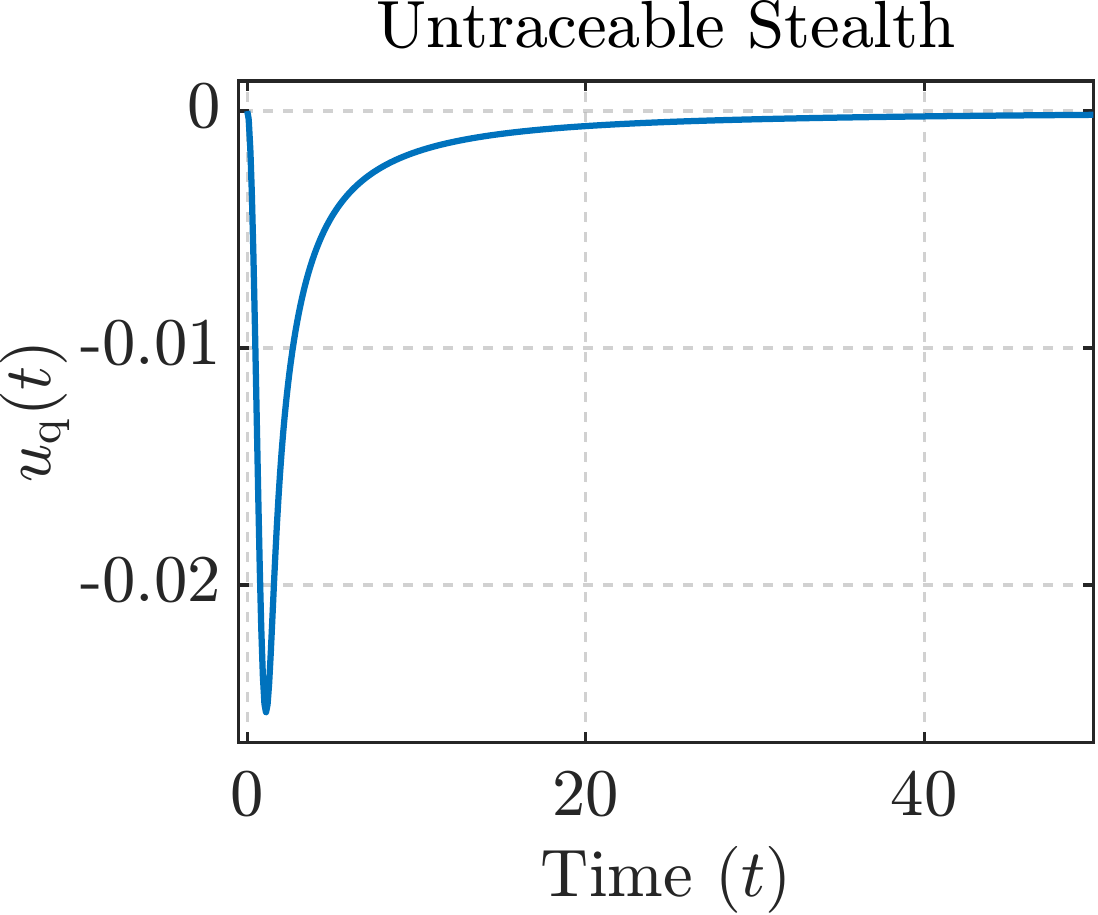}
        \caption{}
        \label{fig: Ex2 uq a 1}
    \end{subfigure}
    % small vertical gap between rows (optional)
    \vspace{1pt}

    % Bottom row: subfigure 3 spanning almost full width
    \begin{subfigure}[t]{0.24\textwidth}
        \centering
        \includegraphics[width=\textwidth,keepaspectratio]{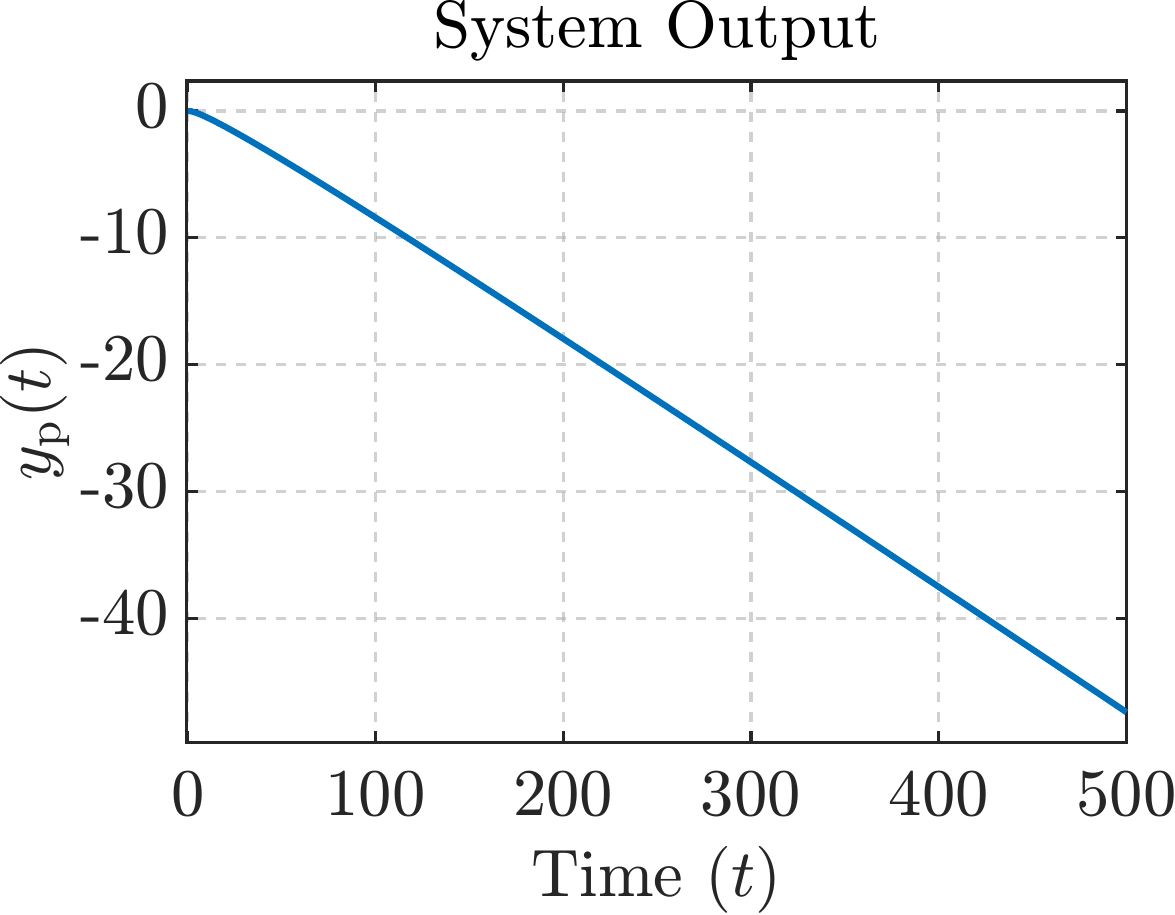}
        \caption{}
        \label{fig: Ex2 yp}
    \end{subfigure}
    \hfill%
    \begin{subfigure}[t]{0.24\textwidth}
        \centering
        \includegraphics[width=\textwidth,keepaspectratio]{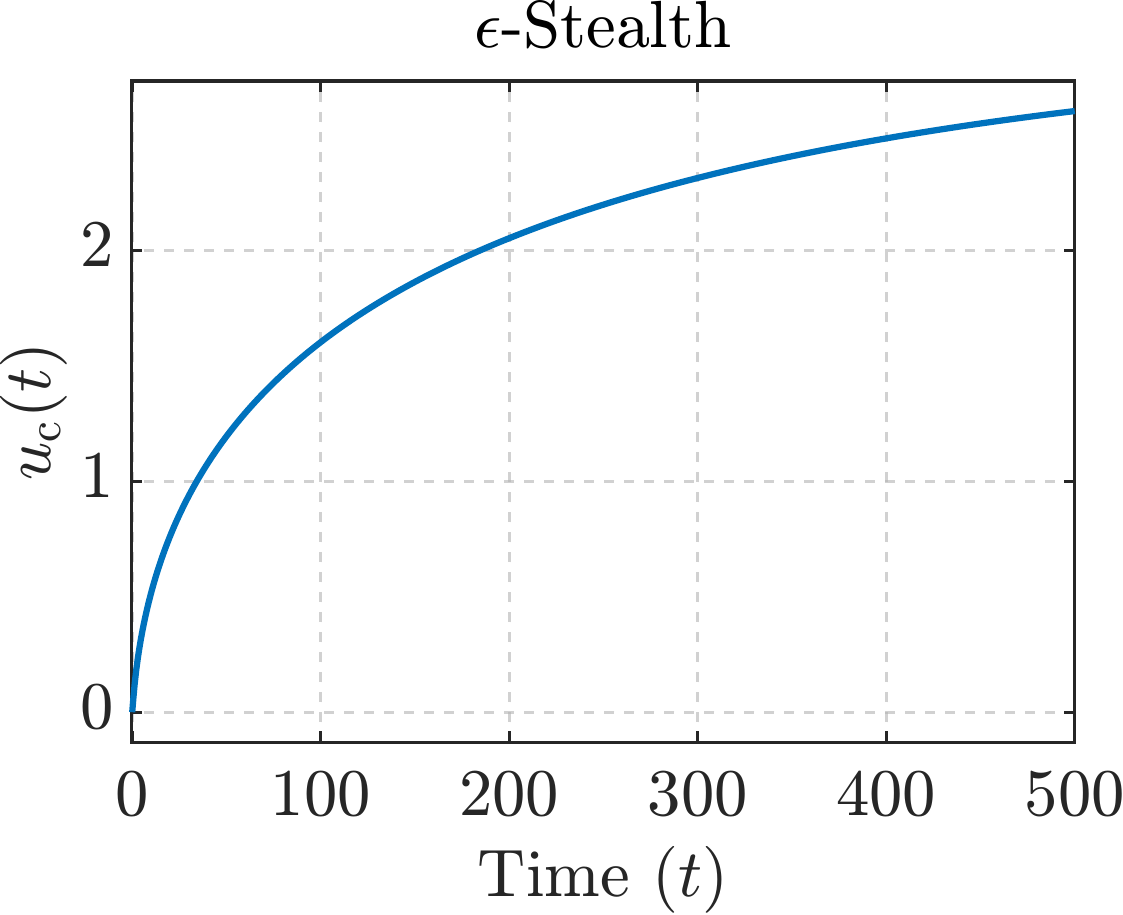}
        \caption{}
        \label{fig: Ex2 uc a 1}
    \end{subfigure}
    %% small vertical gap between rows (optional)
    %\vspace{1pt}
    %% Bottom row: subfigure 3 spanning almost full width
    %\begin{subfigure}[t]{0.24\textwidth}
    %    \centering
    %    \includegraphics[width=\textwidth,keepaspectratio]{Figures/LTI/LTI case 3.pdf}
    %    \caption{}
    %    \label{fig: LTI case 3}
    %\end{subfigure}
    %\hfill%
    %\begin{subfigure}[t]{0.24\textwidth}
    %    \centering
    %    \includegraphics[width=\textwidth,keepaspectratio]{Figures/LTI/LTI attack cubic.pdf}
    %    \caption{}
    %    \label{fig: LTI attack cubic}
    %\end{subfigure}
    %% small vertical gap between rows (optional)
    %\vspace{1pt}
    \caption{Simulation results for Example~2 from Section~\ref{sec: num ex}.} %showing the propagation of the residual signals corresponding to the FDIA signals that grow unbounded over time.} 
    \label{fig: Ex2}
\end{figure}
% Optionally force figures to appear before continuing
%\FloatBarrier

\begin{rem}
Example~2 demonstrates the potential of extending the theoretical results stated in Theorems~\ref{thm: bounded attack} and~\ref{thm: u.s.} to more general settings that permit non-positive-kernel feedback. This further strengthens the argument against the presence of integrators in feedback systems, as they allow for potentially unbounded FDIAs of the form \eqref{eq: FDIA LTV} to remain stealthy.
%exploit these results and remain stealthy, thereby compromising system security. 
%Since $a < q = 2$ and the FDIA is deemed untraceably stealthy with respect to $u_{\s{q}}$, this suggests that the results from Theorems~\ref{thm: bounded attack} and~\ref{thm: u.s.} remain applicable for the non-positive-kernel feedback\footnote{A formal proof of this remains a topic for future work.}. 
\end{rem}

\section{Conclusion} \label{sec: concl}
%In this work, we highlighted the vulnerabilities afforded by the inclusion of integrators in control systems with positive-kernel feedback. In particular, we showed that such systems, with a chain of $q \in \mathbb{Z}_{\geq 1}$ integrators, permit FDIAs of the polynomial form with degree~$a \in \mathbb{Z}_{\geq 0}$ (see \eqref{eq: FDIA LTV}) to remain stealthy with respect to the output of the controller if $a \leq q$. 
In this work, we studied the security of LTV systems that contain $q \in \mathbb{Z}_{\geq 1}$ integrators in the closed-loop under non-negative-kernel feedback. For polynomial FDIA signals of degree $a \in \mathbb{Z}_{\geq 0}$, it was formally proved under the conditions stated in Assumptions~\ref{assm: uq < infty LVIE} and~\ref{assm: uq AS LVIE} that if $a \leq q$ ($a< q$, resp.), then, for some finite $\epsilon > 0$, the FDIA will remain $\epsilon \text{-stealthy}$ (untraceably stealthy, resp.) to an anomaly detector comparing the output of the controller to its expected nominal trajectory.

For future work, we intend to expand %the scope of 
our investigation to include integrator-endowed systems under non-positive-kernel feedback. It was already shown in Example~2 of Section~\ref{sec: num ex} that such systems are susceptible to stealthy FDIAs that can grow unbounded over time. Therefore, this line of inquiry seems promising and, thus, warrants further research.

\bibliography{references}                                                  

@book{burton2005volterra,
  title={Volterra integral and differential equations},
  author={Burton, Theodore Allen},
  volume={202},
  year={2005},
  publisher={Elsevier}
}

@article{tsalyuk1979volterra,
  title={Volterra integral equations},
  author={Tsalyuk, ZB},
  journal={Journal of Soviet Mathematics},
  volume={12},
  number={6},
  pages={715--758},
  year={1979},
  publisher={Springer}
}

@article{Tsa68,
author={Z.~B.~Tsalyuk},
title = {The stability of Volterra equations},
journal={ Differ. Uravn.},
year= {1968},
volume={4},
issue={11},
pages={1967--1979},
url = {http://mi.mathnet.ru/de469}
}

@book{brunner2017volterra,
  title={Volterra integral equations: an introduction to theory and applications},
  author={Brunner, Hermann},
  volume={30},
  year={2017},
  publisher={Cambridge University Press}
}

@book{gripenberg1990volterra,
  title={Volterra integral and functional equations},
  author={Gripenberg, Gustaf and Londen, Stig-Olof and Staffans, Olof},
  number={34},
  year={1990},
  publisher={Cambridge University Press}
}

@book{protter2012intermediate,
  title={Intermediate calculus},
  author={Protter, Murray H and Charles Jr, B and others},
  year={2012},
  publisher={Springer Science \& Business Media}
}

@article{kadelburg2005interchanging,
  title={Interchanging two limits},
  author={Kadelburg, Zoran and Marjanovic, M},
  journal={Enseign. Math},
  volume={8},
  number={1},
  pages={15--29},
  year={2005}
}

@article{pasqualetti2013attack,
  title={Attack detection and identification in cyber-physical systems},
  author={Pasqualetti, Fabio and D{\"o}rfler, Florian and Bullo, Francesco},
  journal={IEEE transactions on automatic control},
  volume={58},
  number={11},
  pages={2715--2729},
  year={2013},
  publisher={IEEE}
}

@article{teixeira2015secure,
  title={A secure control framework for resource-limited adversaries},
  author={Teixeira, Andr{\'e} and Shames, Iman and Sandberg, Henrik and Johansson, Karl Henrik},
  journal={Automatica},
  volume={51},
  pages={135--148},
  year={2015},
  publisher={Elsevier}
}

@article{rao2001naive,
  title={Naive control of the double integrator},
  author={Rao, Venkatesh G and Bernstein, Dennis S},
  journal={IEEE Control Systems Magazine},
  volume={21},
  number={5},
  pages={86--97},
  year={2001},
  publisher={IEEE}
}

@book{franklin2002feedback,
  title={Feedback control of dynamic systems},
  author={Franklin, Gene F and Powell, J David and Emami-Naeini, Abbas and Powell, J David},
  volume={4},
  year={2002},
  publisher={Prentice hall Upper Saddle River}
}

@book{rugh1996linear,
  title={Linear system theory},
  author={Rugh, Wilson J},
  year={1996},
  publisher={Prentice-Hall, Inc.}
}

@techreport{hemsley2018history,
  title={History of industrial control system cyber incidents},
  author={Hemsley, Kevin E and Fisher, E and others},
  year={2018},
  institution={Idaho National Lab.(INL), Idaho Falls, ID (United States)}
}

@article{rudin1976principles,
  title={Principles of mathematical analysis},
  author={Rudin, Walter},
  journal={3rd ed.},
  year={1976}
}
%\begin{thebibliography}{xx}  % you can also add the bibliography by hand

%\bibitem[Able(1956)]{Abl:56}
%B.C. Able.
%\newblock Nucleic acid content of microscope.
%\newblock \emph{Nature}, 135:\penalty0 7--9, 1956.

%\bibitem[Able et~al.(1954)Able, Tagg, and Rush]{AbTaRu:54}
%B.C. Able, R.A. Tagg, and M.~Rush.
%\newblock Enzyme-catalyzed cellular transanimations.
%\newblock In A.F. Round, editor, \emph{Advances in Enzymology}, volume~2, pages
%  125--247. Academic Press, New York, 3rd edition, 1954.

%\bibitem[Keohane(1958)]{Keo:58}
%\newblock \emph{Power and Interdependence: World Politics in Transitions}.
%\newblock Little, Brown \& Co., Boston, 1958.

%\bibitem[Powers(1985)]{Pow:85}
%\newblock Is there a way out?
%\newblock \emph{Harpers}, pages 35--47, June 1985.

%\bibitem[Soukhanov(1992)]{Heritage:92}
%\newblock \emph{{The American Heritage. Dictionary of the American Language}}.
%\newblock Houghton Mifflin Company, 1992.

%\end{thebibliography}

\appendix
\section{Proof of Lemma~\ref{lem: gc tauq < infty}} \label{sec: appndx lemma gc tauq < infty}    % Each appendix must have a short title.
              % Sections and subsections are supported  
                                                                         % in the appendices.
\begin{proof}
\begin{enumerate}[label=\alph*)]
    \item Under Assumption~\ref{assm: uq < infty LVIE}-(a), %it follows from Since \eqref{eq: uq volt} is supposed to be stable per Assumption~\ref{assm: uq < infty LVIE}-(a), 
    it follows from Lemma~\ref{lem: iff stable ZB} that
\begin{align*}
    &\sup_{t \geq 0} \, \int_{0}^t G_{\s{c},\s{p},\s{q}} (t,\sigma) \dd \sigma \nonumber \\ 
    &= \sup _{t \geq 0} \, \int_0^t \, \int_{\sigma}^t G_{\s{c},\s{p}}(t, \tau)\, g_{\s{q}}(\tau -\sigma) \, \dd \tau  \, \dd \sigma  < \infty.    
\end{align*}
Then, by applying Assumption~\ref{assm: uq < infty LVIE}-(c), we obtain the following implication from the inequality above:  
\begin{align} \label{eq: Gcp >= gc}
\sup _{t \geq 0} \, \int_0^t \, \int_{\sigma}^t g_{\s{c}}(t, \tau)\, g_{\s{q}}(\tau -\sigma) \, \dd \tau  \, \dd \sigma  &< \infty. 
\end{align}
Substituting \eqref{eq: gq} in \eqref{eq: Gcp >= gc} and iterating the order of integration then yields
\begin{align*}    
\frac{1}{(q-1)!}\, \sup_{t \geq 0}  \, \int_0^t g_\s{c}(t, \tau) \int_{0}^{\tau} \,(\tau-\sigma)^{q-1} \dd \sigma \, \dd \tau < \infty \\
\iff \frac{1}{q!}\, \sup_{t \geq 0}  \, \int_0^t g_\s{c}(t, \tau) \, \tau^q \, \dd \tau < \infty,
\end{align*}
which implies \eqref{eq: gctauq infty}, and completes the proof. 
%This proves \eqref{eq: gctauq infty}.

\item %Next, we shall prove \eqref{eq: gctaua infty}. 
For $a = q$, \eqref{eq: gctaua bounded} follows trivially from \eqref{eq: gctauq infty}. For $0 \leq a < q\geq 1$, consider the following: For $b_1, b_2 \in \mathbb{R}_{\geq 0}$, $ b_2 \geq b_1$, let
\begin{equation} \label{eq: Gamma proof}
    \Gamma_{t,q}(b_1,b_2) \coloneqq \int_{b_1}^{b_2} g_c(t, \tau) \, \tau^q \, \dd \tau, \quad t \geq 0.
\end{equation}
%Let~$\Gamma_{t,q}(a,b) \coloneqq \int_a^b g_c(t, \tau) \, \tau^q \, \dd \tau$, for $a,b \in \mathbb{R}_{\geq 0}$, $t\geq0$. 
Then, from \eqref{eq: gctauq infty}, it follows that
%\begin{align*}
%\sup_{t \geq 0} \Gamma_{t,q}(0,t) &= \max\biggl\{\sup_{t \in [0,1] }\Gamma_{t,q}(0,t),\, \sup_{t> 1} \Gamma_{t,q}(0,t)\biggr\}< \infty.
%\end{align*}
%Consequently,  
\begin{align}
\sup_{t \in [0,1] }\Gamma_{t,q}(0,t) < \infty&, \label{eq: gammaq sup t01 0t} 
\end{align}
and 
\begin{align}
    \sup_{t > 1} \Gamma_{t,q}(0,t) = \sup_{t > 1} \bigl(\Gamma_{t,q}(0,1) + \Gamma_{t,q}(1,t) \bigr)   < \infty. \label{eq: gammaq sup t>1}
\end{align}
Under Assumption~\ref{assm: uq < infty LVIE}-(c), since $\Gamma_{t,q}(0,1)$ and $\Gamma_{t,q}(1,t) \geq 0$, for all $t \geq 0$, \eqref{eq: gammaq sup t>1} implies
\begin{align} \label{eq: gammaq sup t>1 split}
    \sup_{t>1} \Gamma _{t,q}(0,1) < \infty \text{ and } \sup_{t>1} \Gamma _{t,q}(1,t) < \infty.
\end{align}
Coming back to the proof of \eqref{eq: gctaua bounded} for $a< q$: Since
    \begin{align} \label{eq: gammatq sup t0 0t}
    \sup_{t \geq 0} \Gamma_{t,a}(0,t) &= \max\biggl\{\sup_{t \in [0,1] }\Gamma_{t,a}(0,t),\, \sup_{t> 1} \Gamma_{t,a}(0,t)\biggr\},
    \end{align}
    it suffices to show that both $\sup_{t \in [0,1] }\Gamma_{t,a}(0,t) < \infty$ and $\sup_{t> 1} \Gamma_{t,a}(0,t) < \infty$. 
    To that end, first, we note that $\sup_{t> 1} \Gamma_{t,a}(0,t) = \sup_{t> 1} \bigl(\Gamma_{t,a}(0,1) + \Gamma_{t,a}(1,t) \bigr)$. 
    For $0 \leq a < q\geq 1$, $\tau^a < \tau^q$, for all~$\tau \in \mathbb{R}: \tau > 1$. Consequently, we obtain  
    \begin{equation} \label{eq: gamma t>1 1t}
        \sup_{t>1} \Gamma _{t,a}(1,t) < \sup_{t>1}\Gamma _{t,q}(1,t) < \infty,
    \end{equation}
    where the last inequality follows from \eqref{eq: gammaq sup t>1 split}.
    Now, for~$\tau \in (0,1]$, $0< \tau^q \leq \tau^a \leq 1$, where $0 \leq a < q\geq 1$. This implies that $\tau^a < 1 + \tau^q$, for all $\tau \in [0,1]$. Consequently,
    \begin{align} \label{eq: gammasup ta01}
       \sup_{t> 1}& \Gamma_{t,a}(0,1) <  \sup_{t> 1}\biggl( \int_0^1 g_\s{c}(t,\tau)\, \dd \tau + \Gamma_{t,q}(0,1) \biggr), \nonumber \\
       &\leq \sup_{t> 1}\int_0^1 g_\s{c}(t,\tau)\, \dd \tau + \sup_{t> 1}\Gamma_{t,q}(0,1) < \infty,
    \end{align}
    where the last inequality follows from \mbox{Assumption~\ref{assm: uq < infty LVIE}-(b)} and \eqref{eq: gammaq sup t>1 split}. Together with \eqref{eq: gamma t>1 1t}, this then implies that $\sup_{t>1}\Gamma_{t,a}(0,t) < \infty$ in \eqref{eq: gammatq sup t0 0t}. Following a similar line of reasoning, it could be shown that %for the other term in \eqref{eq: gammatq sup t0 0t}, it follows that
    \begin{align*}
        \sup_{t\in [0,1]}\Gamma _{t,a}(0,t) &< \sup_{t\in [0,1]}\int_0^t g_\s{c}(t,\tau)\, \dd \tau + \sup_{t\in [0,1]}\Gamma_{t,q}(0,t), \\
        &< \infty,
    \end{align*}
    where the last inequality follows from \mbox{Assumption~\ref{assm: uq < infty LVIE}-(b)} and \eqref{eq: gammaq sup t01 0t}. With this, we have shown that for $a \leq q$ $\sup_{t \geq 0}, \Gamma_{t,a}(0,t) < \infty$ in \eqref{eq: gammatq sup t0 0t}, which is equivalent to \eqref{eq: gctaua bounded}. This completes the proof. 
    %Therefore, for $a \leq q$, we have proved $\sup_{t \geq 0} \Gamma_{t,a}(0,t) < \infty$ in \eqref{eq: gammatq sup t0 0t}, which is equivalent to \eqref{eq: gctaua infty}. 
\end{enumerate}
The proof is complete. 
\end{proof}

\section{Proof of Lemma~\ref{lemm: AS}} \label{sec: appndx lemma AS}    % Each appendix must have a short title.
              % Sections and subsectins are supported  
                                                                        
                                                                         % in the appendices.
%To proceed, first, we need the following supporting result. 
Before proceeding with the proof of Lemma~\ref{lemm: AS}, we present the following supporting result. 
\begin{lem} \label{lem: uniform conv}
Under Assumptions~\ref{assm: uq < infty LVIE} and~\ref{assm: uq AS LVIE}, the following limit converges uniformly in $t$
\begin{equation} \label{eq: uniform convergence unproved}
    \lim_{T\to 0^+} \int_T^{t} g_{\s{c}}(t,\tau) (\tau - T)^{q-1} \dd \tau = \int_0^{t} g_{\s{c}}(t,\tau) \tau ^{q-1} \dd \tau, 
\end{equation}
where $q \in \mathbb{Z}_{\geq 1}$.
\end{lem}
\begin{proof}
To complete the proof, it suffices to show that 
\begin{equation} \label{eq: uniform conv}
    \lim_{T \to 0^+ } \sup_{t \geq 0} |F_T(t) - F_0(t) | = 0,
\end{equation}
where $F_T(t) \coloneqq \int_T^t g_{\s{c}}(t,\tau) (\tau - T)^{q -1} \dd \tau$, for $t > T > 0$. 
To that end, it follows that for $q = 1$, (under \mbox{Assumption~\ref{assm: uq < infty LVIE}-(c)}) $|F_{T}(t) - F_0(t)| = \int_0^T g_c(t,\tau) \dd\tau$. Then, 
\eqref{eq: uniform conv} follows by applying \mbox{Assumption~\ref{assm: uq AS LVIE}-(c)}. %Assumption~\ref{assm: uq AS LVIE}-(c).
For $q \in \mathbb{Z}_{\geq 2}$, we note that
\begin{align} \label{eq: FT limit q2}
|&F_{T}(t) - F_0(T) | =\nonumber \\
& \Biggl|\int_T^t g_{\s{c}}(t,\tau)\bigl((\tau - T)^{q-1} - \tau^{q-1}\bigr) \dd \tau - \int_0^Tg_{\s{c}}(t,\tau) \tau^{q-1} \dd \tau\Biggr|, \nonumber \\
&\leq  \int_T^t g_{\s{c}}(t,\tau)  \bigl| \tau^{q-1} -(\tau - T)^{q-1}\bigr| \dd \tau + \int_0^Tg_{\s{c}}(t,\tau) \tau^{q-1} \dd \tau.
\end{align}
For the first term in \eqref{eq: FT limit q2}, by applying the mean value theorem to the continuous function $p(r) \coloneqq  r^{q-1}$ defined over the closed domain $[\tau - T, \tau]$, we obtain %for $q \geq 2$:
\begin{align*} %\label{eq: tau bound}
|\tau^{q-1} -(\tau - T)^{q-1}| &\leq T(q-1) \max_{r \in ((\tau -T), \tau)}r^{q-2}, \\
&\leq T(q-1)\tau^{q-2}.    
\end{align*}
%For~$\tau \in [T,1]$, $\max_{r \in ((\tau -T), \tau)}r^{q-2} \leq 1$. For $\tau \in [1,t]$, $\max_{r \in ((\tau -T), \tau)}r^{q-2} \leq T(q-1)\tau^{q-2}$. Then, splitting the integral in the first term of \eqref{eq: FT limit q2} as $\int_T^t = \int_T^1 + \int_1^t$, and substituting the bound \eqref{eq: tau bound} over appropriate intervals yields $\int_T^tg_{\s{c}}(t,\tau)|\tau^{q-1} -(\tau - T)^{q-1}| \dd \tau\leq T(q-1)\int_T^1 g_{\s{c}}(t,\tau) \dd \tau+ T(q-1)\int_1^tg_{\s{c}}(t,\tau)\tau^{q-2} \dd \tau$. 
For the second term in \eqref{eq: FT limit q2}, for $q \geq 2$, $\tau^{q-1} = \tau \tau^{q-2} \leq T\tau^{q-2}$, for $\tau \in [0,T]$. Substituting these upper bounds (in place of their corresponding terms) in \eqref{eq: FT limit q2}, and then using the resulting upper bound of \eqref{eq: FT limit q2} in \eqref{eq: uniform conv} under \mbox{Assumption~\ref{assm: uq < infty LVIE}-(c)} (which allows us to increase the interval of the integrals while maintaining the upper bound) then yields
\begin{align*}
\lim_{T \to 0^+ } \sup_{t \geq 0} |F_T(t) - F_0(t) | &\leq \lim_{T \to 0^+ } Tq \, \underbrace{\sup_{t \geq 0} \int_0^tg_{\s{c}}(t,\tau) \tau^{q-2}\dd\tau}_{\eqqcolon M < \infty \text{ (per Lemma~\ref{lem: gc tauq < infty})}}, \nonumber \\
&=\lim_{T \to 0^+} TqM = 0.
\end{align*}
This satisfies \eqref{eq: uniform conv} and completes the proof. 
\end{proof}
Now we proceed with the proof of Lemma~\ref{lemm: AS}.
\begin{proof}[Proof of Lemma~\ref{lemm: AS}]
\begin{enumerate}[label=\alph*)]
    \item Firstly, it follows from Lemma~\ref{lem: iff AS ZB} that, for any fixed $T>0$,
\begin{equation} \label{eq: lim Gcpq = 0}
    \lim_{t \to \infty} \int _0^T G_{\s{c}, \s{p}, \s{q}}(t,\sigma) \dd \sigma = 0.
\end{equation}
%It then follows that, for any fixed $T>0$ %$T: \,0< T \leq t$ %\textcolor{red}{}, 
Then, it follows from the definition of $G_{\s{c},\s{p}, \s{q}}$ in \eqref{eq: uq volt} and \mbox{Assumption~\ref{assm: uq < infty LVIE}-(c)} that
\begin{align}
    &\lim_{t \to \infty} \int_{0}^T G_{\s{c}, \s{p}, \s{q}}(t,\sigma) \dd \sigma, \nonumber \\
    &\geq \lim_{t\to \infty} \int_0^T \biggl( \int_{\sigma}^t g_\s{c}(t,\tau)\, g_{\s{q}}(\tau-\sigma) \, d\tau \biggr) \, \dd \sigma, \nonumber \\
    &=  \frac{1}{(q-1)!} \lim_{t\to \infty} \Biggl(\int_0^T g_\s{c}(t, \tau)\biggl( \int_0^\tau  (\tau-\sigma)^{q-1} \, \dd \sigma \biggr) d\tau \nonumber \\ 
    &\phantom{\quad\quad}
    + \int_T^t g_\s{c}(t,\tau)\biggl( \int_0^T (\tau-\sigma)^{q-1} \dd \sigma \biggr) \dd \tau \Biggr),  \label{eq: odd iteration itegration} \\
    &=  \frac{1}{q!} \lim_{t\to \infty} \Biggl( \int_0^T g_\s{c}(t, \tau) \tau^q \, \dd \tau \nonumber \\
    & \phantom{\frac{1}{q!} \quad\quad} 
    + \int_T^t g_\s{c}(t, \tau) \bigl( \tau^q -(\tau - T)^q \bigr) \, \dd \tau \Biggr)   \geq  0,  \label{eq: the expression}
\end{align}
where the equality in \eqref{eq: odd iteration itegration} follows from the iteration of the order of integration (allowed under the continuity of the integrand, where continuity of the integrand follows from Lemma~\ref{lem: g cont}). Furthermore, the last inequality ($\geq 0$) in \eqref{eq: the expression} follows from the observation that each of the terms inside the main parentheses in \eqref{eq: the expression} is non-negative per Assumption~\ref{assm: uq < infty LVIE}-(c). Finally, we note that this last inequality implies equality ($= \!0$), as the expression in \eqref{eq: the expression} cannot be greater than, and less than, $0$ simultaneously. %, unless it is equal to~$0$. 
Substituting~$\mu \coloneqq \tau - T$ in \eqref{eq: the expression}, and subsequently applying the binomial expansion 
\begin{align} \label{eq: binom}
    (\mu+T)^q - \mu^q = qT\mu^{q-1} + \sum_{k=2}^q \binom{q}{k} \mu^{q-k} T^k,
\end{align}
where $\binom{q}{k}\coloneqq 0$ for $ k > q \geq 1$, yields

%the following implication is obtained from \eqref{eq: the expression}:
%\begin{equation} \label{eq: the new expression}
%\begin{split}
%\lim_{t\to \infty} \Biggl( &\int_0^T g_\s{c}(t, \tau)\tau^q \, \dd \tau  \\
%&+ \int_T^t g_{\s{c}}(t,\tau) \bigl(\mu + T)^q - \mu^q\bigr) \dd \tau \biggr) = 0.
%\end{split}
%\end{equation}
%To proceed, we define $\binom{q}{k}\coloneqq 0$ for $ k > q \geq 1$, thereby permitting the following binomial expansion
%\begin{align} \label{eq: binom}
%    (\mu+T)^q - \mu^q = qT\mu^{q-1} + \sum_{k=2}^q \binom{q}{k} \mu^{q-k} T^k.
%\end{align}
%Then, we substitute \eqref{eq: binom} into \eqref{eq: the new expression}, yielding
%% \begin{align} \label{eq: parent lim implication}
%%    \lim_{t\to \infty} \Biggl( \int_0^T g_\s{c}(t, \tau) \tau^q \, \dd \tau + \int_T^t g_\s{c}(t, \tau) \bigl( \tau^q -(\tau - T)^q \bigr) \dd \tau \Biggr) =0.
%%\end{align}
\begin{align} \label{eq: after mu and bin}
    &  \lim_{t\to \infty} \Biggl( \int_0^T g_\s{c}(t, \tau) \tau^q \, \dd \tau + qT\int_0^{t-T} \!\!\!\! g_{\s{c}}(t,\mu + T) \mu^{q-1} \dd \mu \nonumber \\
    &+ \sum_{k=2}^q \binom{q}{k} T^k \int_0^{t-T} g_\s{c}(t, \mu + T) \mu^{q-k} \dd \mu \Biggr) =0.
\end{align}
Since each of the terms inside the parentheses is non-negative (under Assumption~\ref{assm: uq < infty LVIE}-(c)), %inside the  \eqref{eq: after mu and bin} is non-negative (under Assumption~\ref{assm: uq < infty LVIE})
and the limit of the sum exists (as given in \eqref{eq: after mu and bin}), the sum of the limits (the limit in \eqref{eq: after mu and bin} when distributed inside the parentheses and applied to each term) exists, and each of them (the limits) also converges to $0$. Keeping this in mind, we focus on the second term inside the parentheses in \eqref{eq: after mu and bin}, re-substitute $\mu = \tau - T$, divide it by $T>0$, and take the limit~$T \to 0^+$ to obtain %(after substituting \ back) % the following  second term
\begin{equation} \label{eq: limit before iteration}
    \lim_{T\to 0^+}\lim_{t\to \infty} \int_T^{t} g_{\s{c}}(t,\tau) (\tau - T)^{q-1} \dd \tau = 0.
\end{equation}
%\textcolor{red}{Assuming the following limit converges uniformly in~$t$,}
%Now, if the following limit converges uniformly in $t$,
%Then, to complete the proof, it suffices to show that the following limit converges uniformly in $t$
%\begin{equation} \label{eq: uniform convergence unproved}
%    \lim_{T\to 0^+} \int_T^{t} g_{\s{c}}(t,\tau) (\tau - T)^{q-1} \dd \tau = \int_0^{t} g_{\s{c}}(t,\tau) \tau ^{q-1} \dd \tau, 
%\end{equation}
%as then we can employ \cite[Theorem~1]{kadelburg2005interchanging} to iterate the order of the limit in \eqref{eq: limit before iteration} and apply \eqref{eq: uniform convergence unproved} to obtain \eqref{eq: gctauq-1 AS}, thereby completing the proof. %Therefore, it suffices to show that the limit in \eqref{eq: uniform convergence unproved} converges uniformly in $t$. 
%and note that the uniform convergence of \eqref{eq: uniform convergence unproved} in $t$ is implied by the following condition
To complete the proof, it suffices to show that the following limit 
\begin{equation} \label{eq: uniform convergence proved}
    \lim_{T\to 0^+} \int_T^{t} g_{\s{c}}(t,\tau) (\tau - T)^{q-1} \dd \tau = \int_0^{t} g_{\s{c}}(t,\tau) \tau ^{q-1} \dd \tau, 
\end{equation}
converges uniformly in $t$, as it would then allow us to employ \cite[Theorem~1]{kadelburg2005interchanging} to iterate the order of the limit in \eqref{eq: limit before iteration} and apply \eqref{eq: uniform convergence proved} to obtain \eqref{eq: gctauq-1 AS}, thereby completing the proof. %Therefore, it suffices to show that the limit in \eqref{eq: uniform convergence unproved} converges uniformly in $t$. 
Lemma~\ref{lem: uniform conv} establishes \eqref{eq: uniform convergence proved}, which completes the proof. 
\item From \eqref{eq: phi}, we note that if $a = q-1$, then \eqref{eq: gctaua AS phi} follows trivially. For the case where $q > 1$ and $0 \leq a < q-1$, %consider the following: For all $\tau \in \mathbb{R}: \tau >1$, $\tau^a < \tau^{q-1}$. Furthermore, for $ \tau \in (0, 1]$, $ 0 < \tau^{q-1 }\leq \tau^a \leq 1$. Then, 
it follows from \eqref{eq: gctauq-1 AS} that, for $t \geq 1$, 
\begin{align} 
&\lim_{t \to \infty} \int_0^t g_c(t,\tau) \tau^a \dd\tau, \nonumber \\
&= \lim_{t \to \infty} \biggl(\int_0^1 g_\s{c}(t,\tau) \tau^a \dd \tau + \int_1^t g_\s{c} (t,\tau) \tau^a \dd \tau \biggr), \label{eq: gctaua limit = 0 eqq} \\
%& = \lim_{t \to \infty} \int_0^1 g_\s{c}(t,\tau) \tau^a \dd \tau + \lim_{t \to \infty}\int_1^t g_\s{c} (t,\tau) \tau^a \dd \tau, \nonumber \\
&\leq \lim_{t \to \infty} \int _0^1 g_\s{c}(t, \tau) \dd \tau + \lim_{t \to \infty} \int_0^t g_\s{c}(t, \tau) \tau^{q-1} \dd\tau = 0, \label{eq: gctaua limit = 0} % \\%\nonumber \\
%&\phantom{\leq} + \lim_{t \to \infty} \int_1^t g_c(t,\tau) \tau^{q-1} \dd \tau = 0,
%&=0,\label{eq: gctaua limit = 0 equality}
\end{align}
where the last inequality in \eqref{eq: gctaua limit = 0} is explained by observing the following: i) each term in the limit of sums in \eqref{eq: gctaua limit = 0 eqq} is non-negative under Assumption~\ref{assm: uq < infty LVIE}-(c), therefore, the sum of limits exists and, thus, the limit can be distributed inside the parentheses. ii) Given that $0 \leq a < q-1 \geq 1$, we have $\tau^a < \tau^{q-1}$, for all $\tau  > 1$. Furthermore, for $\tau \in (0,1]$, $0< \tau^{q-1} \leq \tau^a \leq 1$. %, where $0 \leq a < q-1 \geq 1$. 
This implies that $\tau^a < 1 + \tau^{q-1}$, for all $\tau \in [0,1]$. Substituting these upper bounds of $\tau^a$ in \eqref{eq: gctaua limit = 0 eqq} (over appropriate intervals of integration), and observing that $\int_0^1g_c(t,\tau)\tau^{q-1} \dd\tau + \int_1^tg_c(t,\tau)\tau^{q-1} \dd\tau  = \int_0^tg_c(t,\tau)\tau^{q-1} \dd\tau$, then yields the inequality in \eqref{eq: gctaua limit = 0}. Finally, the equality in \eqref{eq: gctaua limit = 0} follows from Assumption~\ref{assm: uq AS LVIE}-(b) and \eqref{eq: gctauq-1 AS}. That is, \eqref{eq: gctaua limit = 0} implies $ \lim_{t \to \infty} \int_0^t g_c(t,\tau)\tau^a \dd\tau \leq 0$. Under Assumption~\ref{assm: uq < infty LVIE}-(c), the inequality here then implies equality ($=\!0$), which yields \eqref{eq: gctaua AS phi} and completes the proof. 
\end{enumerate}
The proof is complete.
\end{proof} 
\end{document}